\documentclass[11pt]{amsart}
\pdfoutput=1

\usepackage{amssymb}
\usepackage{amsthm}
\usepackage{amsfonts}
\usepackage{amsmath}
\usepackage{pmboxdraw}
\usepackage{verbatim} 
\usepackage{graphicx}
\usepackage{color}
\usepackage[colorlinks=true, citecolor=blue, filecolor=black, linkcolor=black, urlcolor=black]{hyperref}
\usepackage{cite}
\usepackage[normalem]{ulem}
\usepackage{subcaption}
\usepackage{bbm}
\usepackage{bm}
\usepackage{mathtools}
\mathtoolsset{showonlyrefs}
\usepackage{todonotes}

\newcommand{\RN}[1]{%
	\textup{\uppercase\expandafter{\romannumeral#1}}%
}

\def\bp{{\bar\partial}}

\def\pa{\partial}

\def\wh{\widehat}
\def\wt{\widetilde}

\def\SS{\mathcal{S}}

\def\C{\mathbb{C}}

\def\P{\mathbf{P}}
\def\R{\mathbb{R}}

\newcommand{\erfc}{\operatorname{erfc}}

\newcommand{\re}{\operatorname{Re}}

\newcommand{\Ext}{\operatorname{Ext}}
\newcommand{\Int}{\operatorname{Int}}

\usepackage{tikz}
\usetikzlibrary{arrows}
\usetikzlibrary{decorations.pathmorphing}
\usetikzlibrary{decorations.markings}
\usetikzlibrary{patterns}
\usetikzlibrary{automata}
\usetikzlibrary{positioning}
\usepackage{tikz-cd}
\tikzset{->-/.style={decoration={
			markings,
			mark=at position #1 with {\arrow{latex}}},postaction={decorate}}}

\tikzset{-<-/.style={decoration={
			markings,
			mark=at position #1 with {\arrowreversed{latex}}},postaction={decorate}}}

\usetikzlibrary{shapes.misc}\tikzset{cross/.style={cross out, draw, 
		minimum size=2*(#1-\pgflinewidth), 
		inner sep=0pt, outer sep=0pt}}

\usepackage{pgfplots}

\definecolor{dullmagenta}{rgb}{0.4,0,0.4}   
\definecolor{darkblue}{rgb}{0,0,0.4}
\hypersetup{linkcolor=black,citecolor=blue,filecolor=dullmagenta,urlcolor=darkblue}

\theoremstyle{plain}
\newtheorem*{thm*}{Theorem}
\newtheorem{thm}{Theorem}[section]
\newtheorem{lem}[thm]{Lemma}

\newtheorem{prop}[thm]{Proposition}
\newtheorem*{prop*}{Proposition}
\newtheorem*{lem*}{Lemma}

\theoremstyle{definition}
\newtheorem*{eg*}{Example}
\newtheorem*{egs*}{Examples}
\newtheorem*{def*}{Definition}
\newtheorem*{Q*}{Question}

\theoremstyle{remark}
\newtheorem*{rmk*}{Remark}
\newtheorem*{rmks*}{Remarks}

\numberwithin{equation}{section}

\begin{document}
\title[Determinantal Coulomb gas ensembles of lemniscate archipelago type]{Determinantal Coulomb gas ensembles with a class of discrete rotational symmetric potentials }

\author{Sung-Soo Byun}
\address{Center for Mathematical Challenges, Korea Institute for Advanced Study, 85 Hoegiro, Dongdaemun-gu, Seoul 02455, Republic of Korea}
\email{sungsoobyun@kias.re.kr}

\author{Meng Yang}
\address{Department of Mathematical Sciences, University of Copenhagen, Copenhagen, Universitetsparken 5, 2100 København Ø, Denmark}
\email{my@math.ku.dk}

\date{\today}

\thanks{Sung-Soo Byun was supported by Samsung Science and Technology Foundation (SSTF-BA1401-51), by the National Research Foundation of Korea (NRF-2019R1A5A1028324) and by a KIAS Individual Grant (SP083201) via the Center for Mathematical Challenges at Korea Institute for Advanced Study. 
Meng Yang was supported by VILLUM FOUNDEN research grant no. 29369 and the NNF grant NNF20OC0064428.}

\begin{abstract}
We consider determinantal Coulomb gas ensembles with a class of discrete rotational symmetric potentials whose droplets consist of several disconnected components. 
Under the insertion of a point charge at the origin, we derive the asymptotic behaviour of the correlation kernels both in the macro- and microscopic scales.
In the macroscopic scale, this particularly shows that there are strong correlations among the particles on the boundary of the droplets. 
In the microscopic scale, this establishes the edge universality. 
For the proofs, we use the nonlinear steepest descent method on the matrix Riemann-Hilbert problem to derive the asymptotic behaviours of the associated planar orthogonal polynomials and their norms up to the first subleading terms. 
\end{abstract}

\maketitle

\section{Introduction and main results}
We consider a configuration $\{z_j\}_{1}^N$ of $N$ points in $\C$ with joint probability distribution 
\begin{equation} \label{Gibbs}
 d \P_N =\frac{1}{Z_N} \prod_{j>k} |z_j-z_k|^2 \prod_{j=1}^N e^{-N Q(z_j)}\, dA(z_j), \qquad dA(z):=\frac{d^2z}{\pi},
\end{equation}
where $Z_N$ is the normalisation constant and $Q:\C \to \R$ is a suitable function called external potential. 
The ensemble \eqref{Gibbs} corresponds to the eigenvalue system of the random normal matrix model, which can be interpreted as the two-dimensional Coulomb gas ensemble at a specific inverse temperature $\beta=2$.
For a recent account of the theory and various topics on the Coulomb gas ensemble, we refer the reader to \cite{Lewin22} and references therein. 

By definition, the $k$-point correlation function $R_{N,k}$ of the system \eqref{Gibbs} is given by 
\begin{equation} \label{RNk}
	R_{N,k}(z_1,\cdots,z_k):= \frac{N!}{(N-k)!} \int_{ \C^{n-k} }    \P_N  \prod_{j=k+1}^{N} dA(z_j). 
\end{equation}
The normalised $1$-point function $\frac{1}{N}R_{N,1}$ corresponds to the macroscopic density of the model.
It is well known that as $N \to \infty$, the empirical measure of $\{ z_j \}_1^N$ converges to Frostman's equilibrium measure, see e.g. \cite{MR3820329,MR4244340}. 
In particular, the system $\{ z_j \}_1^N$ tends to occupy certain compact set $S$ called the droplet. 

The $k$-point function $R_{N,k}$ can be effectively analysed in terms of the correlation kernel. 
To be more concrete, let $p_k \equiv p_{k,N}$ be the $k$:th orthonormal polynomial with respect to the weighted Lebesgue measure $e^{-NQ}\,dA$:
\begin{equation} \label{OP general}
\int_\C p_j(z)\overline{ p_k(z) } e^{-NQ(z)}\,dA(z)=\delta_{jk}, 
\end{equation}
where $\delta_{jk}$ is the Kronecker delta.
We write
\begin{equation} \label{KN ONP}
    K_N(z,w) =e^{-\frac{N}{2} (Q(z) + Q(w))  } \sum_{j=0}^{N-1} p_{j}(z)\overline{p_{j}(w)} 
\end{equation}
for the weighted reproducing kernel of analytic polynomials (of degree less than $N-1$) in $L^2(e^{-NQ} \,dA)$.
Then the $k$-point function $R_{N,k}$ in \eqref{RNk} is expressed as
\begin{equation} \label{RNk det}
R_{N,k}(z_1,\cdots,z_k)= \det \Big[  K_N(z_j,z_l) \Big]_{j,l=1}^k.
\end{equation}
We mention that the correlation kernel can be defined up to a sequence of cocycles, i.e. 
\begin{equation}
\det \Big[  K_N(z_j,z_l) \Big]_{j,l=1}^k=\det \Big[ g_N(z_j) \, \overline{ g_N(z_l)} \cdot K_N(z_j,z_l) \Big]_{j,l=1}^k, 
\end{equation}
where $g_N$ is a continuous unimodular function.

Due to the property \eqref{RNk det}, the system \eqref{Gibbs} is also called the determinantal Coulomb gas ensemble.
Moreover, this naturally calls for the investigation of various asymptotic behaviours of $K_N$ as $N\to\infty$.
Here, one has to distinguish two cases, the \textbf{macroscopic} scale and the \textbf{microscopic} scale. 

\medskip 

The asymptotic behaviour in the microscopic scale is closely related to the universality principle in random matrix theory. 
To describe the local statistics of the model at a given base point $p\in S$, one needs to investigate the asymptotic behaviour of the function
\begin{equation} \label{K rescaling}
(z,w) \mapsto K_N\Big( p+\frac{e^{i\theta}\, z}{\sqrt{N\Delta Q(p)}}\, , p+\frac{e^{i\theta} \, w}{\sqrt{N\Delta Q(p)}} \Big).
\end{equation}
Here if $p \in \pa S$, the angle $\theta \in [0,2\pi)$ is chosen so that $e^{i\theta}$ is outer normal to $\pa S$ at $p$, and otherwise $\theta=0.$
We remark that the specific choice of the rescaling factor $\sqrt{N \Delta Q(p)}$ in \eqref{K rescaling} (which is often called the ``unfolding'') comes from the fact that $\frac{1}{N}R_{N,1}(p) \sim \Delta Q(p)$.

For the bulk case when $p \in \Int S$, it was shown in \cite{ameur2011fluctuations} that for a general external potential $Q$,
\begin{equation} \label{universality bulk}
K_N\Big( p+\frac{z}{\sqrt{N\Delta Q(p)}}\, , p+\frac{w}{\sqrt{N\Delta Q(p)}} \Big) \to G(z,w):= e^{z \bar{w}-\frac{|z|^2}{2}-\frac{|w|^2}{2} }.
\end{equation}
Here $\Int S$ stands for the interior of $S$, the largest open set of $S$, and the universal scaling limit $G$ in \eqref{universality bulk} is called the Ginibre kernel \cite{ginibre1965statistical}. 
For the edge case when $p \in \pa S$, it was shown in a fairly recent work \cite{hedenmalm2017planar} that for a general external potential $Q$,  
\begin{equation} \label{universality edge}
K_N\Big( p+\frac{e^{i\theta}\, z}{\sqrt{N\Delta Q(p)}}\, , p+\frac{e^{i\theta} \, w}{\sqrt{N\Delta Q(p)}} \Big) \to G(z,w)\, \frac12  \erfc\Big( \frac{z+\bar{w}}{\sqrt{2}} \Big).
\end{equation}
The class of potentials $Q$ covered in \cite{hedenmalm2017planar} is quite general but dependent on the topology of the associated droplet.

\medskip

Turning to the macroscopic scale, recently, Ameur and Cronvall \cite{ameur2021szego} made significant results on the asymptotic behaviour of $K_N(z,w)$. 
For the Ginibre ensemble with $Q(z)=|z|^2$, they obtained a precise asymptotic result. 
Namely, it was obtained in \cite[Theorem 1.1]{ameur2021szego} that 
\begin{equation} \label{KN Ginibre Szego}
 K_{N}(z,w) = \sqrt{ \frac{N}{2\pi} } \frac{1}{z\bar{w}-1} (z\bar{w})^N e^{N -\frac{N}{2}(|z|^2+|w|^2) } \cdot  \Big( 1+O(\frac{1}{N}) \Big),
\end{equation}
where $z \not= w$ and $z \bar{w}$ is outside the Szeg\H o curve
\begin{equation} \label{S1 szego}
\mathcal{S}_1:=\{ z \in \C: |z| \le 1, |z \, e^{1-z}|=1 \}.
\end{equation} 
Here, we intentionally add the subscript $1$ since \eqref{S1 szego} can be realised as a special case of $\SS_a$ in \eqref{Sa a>1} below with $a=1$. 
We stress that \cite[Theorem 1.1]{ameur2021szego} indeed provides a closed form of large-$N$ expansions of $K_N.$ 
Let us also mention that \eqref{KN Ginibre Szego} can also be interpreted as an asymptotic result of the incomplete gamma function with complex argument, see \cite[Section 1.4]{ameur2021szego} and \eqref{Q asymp outer}. (Cf. this was crucially used in a recent work \cite{byun2022almost}.)

Beyond the Ginibre ensemble, Ameur and Cronvall considered general external potential $Q$ and derived the uniform asymptotic behaviour of $K_N(z,w)$ for $z,w$ outside the droplet, see \cite[Theorem 1.3]{ameur2021szego}. 
(We also refer to \cite{ADM,MR1371262,molag2022edge} for similar results on the elliptic Ginibre ensemble.)
In particular, they showed that there are strong correlations among the particles on the boundary of the droplet. 
One of the main ingredients in their proof is the asymptotic behaviour of planar orthogonal polynomials \eqref{OP general} due to Hedenmalm and Wennman \cite{hedenmalm2017planar}. 

\medskip 

The above-mentioned results were mainly obtained for the case where the external potential $Q$ is fixed, i.e. independent of $N$.
Nevertheless, the case when $Q$ depends on $N$ is also interesting in particular in the context of the insertion of point charges \cite{ameur2018random} also known as the induced ensembles \cite{MR2881072} or spectral singularities \cite{MR2020225}. 
(Another important example that $N$-dependence of the potential being crucial is the almost-Hermitian regime, see e.g. \cite{AB}.)

Furthermore, in \cite{hedenmalm2017planar} (and also in the follow-up paper \cite{hedenmalm2021soft}), the asymptotic behaviours of planar orthogonal polynomials were constructed in terms of a conformal map from the outside the droplet onto the outside the unit disc. 
Accordingly, the asymptotic result in \cite{hedenmalm2017planar} was obtained for the potential $Q$ whose associated droplet is simply connected as a domain on the Riemann sphere. 
As a consequence, the edge universality \eqref{universality edge} in \cite{hedenmalm2017planar} as well as the Szeg\H o type asymptotic behaviour in \cite{ameur2021szego} were obtained under the assumption that the associated droplet does not have several disconnected components. 

In this work, we aim to provide concrete examples of asymptotic results for the ensembles with a class of $N$-dependent potentials associated with disconnected droplets, see Figure~\ref{Fig:LemAZ}.

\begin{figure}[h!]
    \centering
    \includegraphics[width=0.8\textwidth]{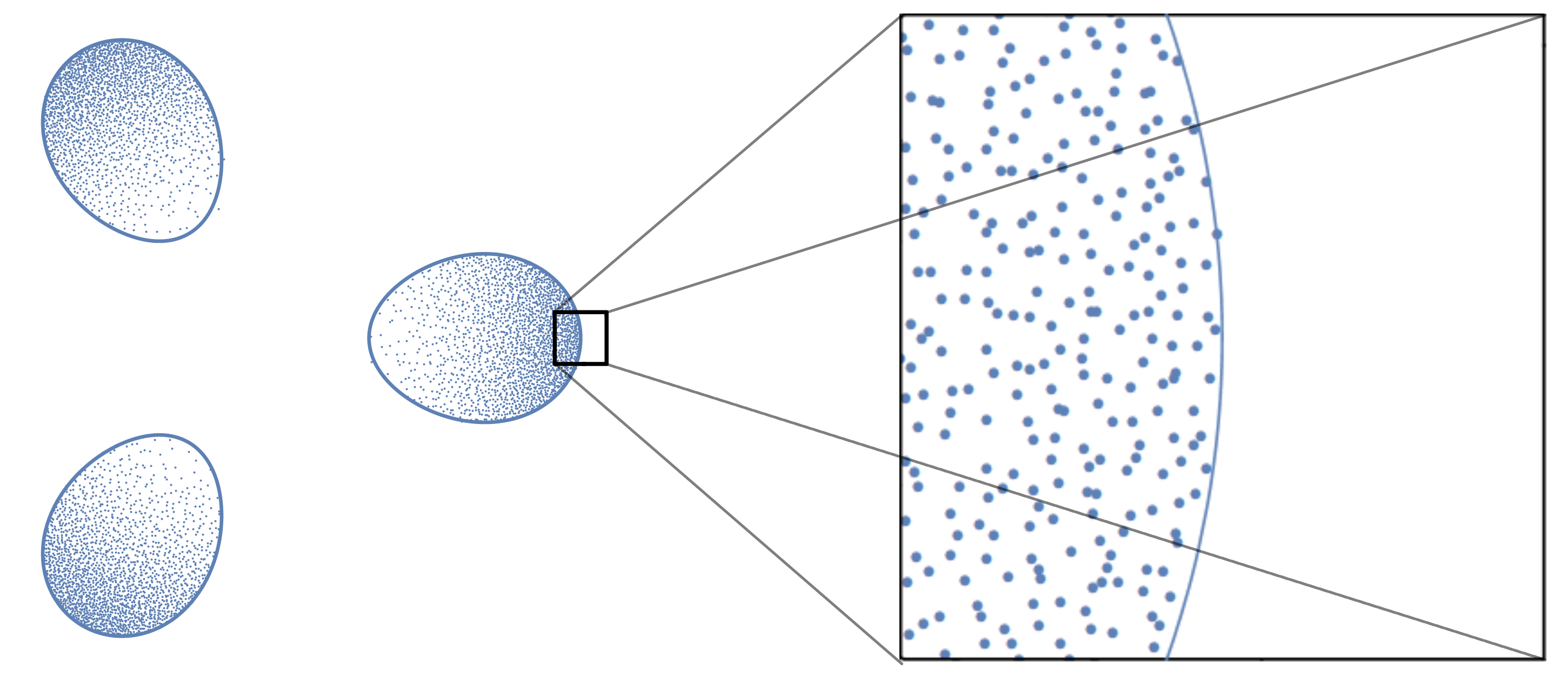}
    \caption{Illustration of the lemniscate archipelago and zooming process}
    \label{Fig:LemAZ}
\end{figure}

\subsection{Main results}

We now precisely introduce our models. 
It is more convenient to begin with a special case when removing the discrete rotational symmetry. 
In this case, the model corresponds to the induced Ginibre ensemble \cite{MR2881072} with the potential
\begin{equation} \label{Q induced Gin}
Q_c(z) \equiv Q_{N,c}(z):=|z|^2-\frac{2c}{N}\log|z-a|, 
\end{equation}
where $c>-1$ and $a \ge 0$. 
From the statistical physics point of view, we insert a point charge $c$ at a given point $a$. 
When $c$ is an integer, the ensemble \eqref{Gibbs} with the potential \eqref{Q induced Gin} can also be realised as the Ginibre ensemble conditioned to have eigenvalue $a$ with multiplicity $c$.

The orthogonal polynomials associated with \eqref{Q induced Gin} reveal a discontinuity at $c=0$. Namely, if $c=0$, since the orthogonal polynomials are simply given by monomials, all the zeros are located at the origin. On the other hand, in \cite{MR3670735}, it was shown that for any $c \not= 0$ and $a>1$, the zeros of orthogonal polynomials tend to occupy the \emph{limiting skeleton} (also known as \emph{mother body}, cf. \cite{MR3289140})
\begin{equation} \label{Sa a>1}
\mathcal{S}_a := \Big \{ z \in \C: \log |z|-a \re z=\log \Big(\frac{1}{a}\Big)-1 \, , \, \re z \le \frac{1}{a} \Big \}. 
\end{equation}
Note that $\SS_a$ crosses the point $1/a$.
The limiting skeleton $\SS_a$ plays an important role in the asymptotic behaviours of the orthogonal polynomials. 
See Figure~\ref{Figure_Sd1} for the shape of $\SS_a$.

\begin{figure}[h!]
	\begin{subfigure}{0.32\textwidth}
		\begin{center}	
			\includegraphics[width=\textwidth]{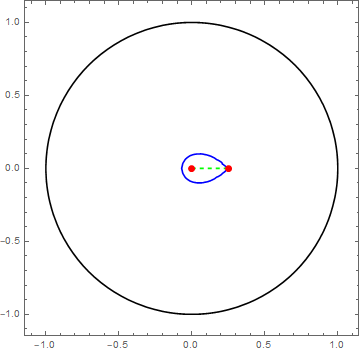}
		\end{center}
		\subcaption{$a=4$}
	\end{subfigure}	
	\begin{subfigure}{0.32\textwidth}
		\begin{center}	
			\includegraphics[width=\textwidth]{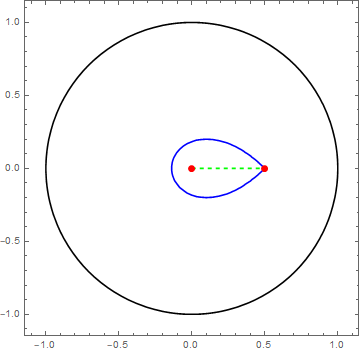}
		\end{center}
		\subcaption{$a=2$}
	\end{subfigure}	
	\begin{subfigure}[h]{0.32\textwidth}
		\begin{center}
			\includegraphics[width=\textwidth]{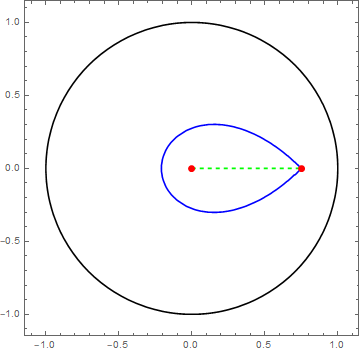}
		\end{center} \subcaption{$a=4/3$}
	\end{subfigure}
	\caption{ The plots display $\SS_a$. The red dots show the origin and $1/a$. The green dashed lines indicate the branch cuts in Theorem~\ref{Thm_induced Ginibre}. } \label{Figure_Sd1}
\end{figure}

In our first result, we obtain the following asymptotic behaviour of $K_N$ in the macroscopic scaling. 

\begin{thm} \label{Thm_induced Ginibre} \textup{\textbf{(Macroscopic asymptotic of the induced Ginibre ensemble)}}
Let $Q$ be the induced Ginibre potential \eqref{Q induced Gin} with $a>1$ and $c>-1$ $(c \not =0)$. Suppose that $z$ and $w$ are outside $\SS_a$, and $|z-w| > \delta $ for some $\delta >0$. Then we have 
\begin{equation} \label{KN induced Szego}
 K_{N}(z,w)  =  \sqrt{ \frac{N}{2\pi} } \frac{1}{z\bar{w}-1}\,\Big( \frac{z}{1-az}  \frac{\bar{w}}{ 1-a\bar{w} } \Big)^c  (z\bar{w})^{N} 
 |(z-a)(w-a)|^c  e^{ N-\frac{N}{2}(|z|^2+|w|^2) } \cdot  \Big(1+O(\frac{1}{N})\Big).
\end{equation}
Here the branch cuts for the variables $z$ and $\bar{w}$ are the line segment $[0,1/a]$.
\end{thm}

Note that if we formally put $c=0$, the formula \eqref{KN induced Szego} corresponds \eqref{KN Ginibre Szego}. 
We mention that the condition $z$ and $w$ being outside the limiting skeleton was also considered in \cite{ADM} for the elliptic Ginibre ensemble. 
(In this case, the limiting skeleton is a line segment connecting two foci of the ellipse.)

In the spirit of the edge universality \eqref{universality edge}, we obtain the following.

\begin{thm}\label{Thm_Boundary iGinibre} \textup{\textbf{(Boundary scaling limits of the induced Ginibre ensemble)}}
Let $Q$ be the induced Ginibre potential \eqref{Q induced Gin} with $a>1$ and $c>-1$. 
Let $p$ be a point on the unit circle. 
Then as $N \to \infty$, we have 
\begin{equation}
\frac{1}{N} K_N\Big(p+\frac{p\,z}{\sqrt{N}}\,,\, p+\frac{p\, w}{\sqrt{N}}\Big) \to G(z,w)\, \frac12 \erfc\Big( \frac{z+\bar{w}}{\sqrt{2}} \Big),
\end{equation}
uniformly for $z,w$ on compact subsets of $\C$.
\end{thm}

We now discuss the ensemble with discrete rotational symmetry.
For $a \ge 0$ and $d \in \mathbb{N}$, let 
\begin{equation} \label{V Vc}
V(z)=\frac{1}{d}|z^d-a|^2, \qquad V_c(z)\equiv V_{N,c}:=V(z)-\frac{2c}{N}\log|z|.
\end{equation}
We refer to \cite{MR4030288,byun2022characteristic, charlier2021asymptotics1} and references therein for recent studies on such models.
Note that the induced Ginibre potential \eqref{Q induced Gin} corresponds to \eqref{V Vc} with $d=1$ up to a translation.  
It is well known that the droplet $S_V$ associated with the potential $V$ is given by
\begin{equation} \label{SV droplet}
S_V:= \{ z\in \C: |z^d-a| \le 1 \},
\end{equation}
see e.g. \cite[Lemma 1]{balogh2015equilibrium}.
The density with respect to $dA$ is given by
\begin{equation} \label{density lemniscate}
 \Delta V(z)=d\,|z|^{2d-2}. 
\end{equation}
Due to the explicit formula \eqref{SV droplet}, one can easily notice that if $a<1$, $S_V$ is connected. 
On the other hand, if $a>1$, $S_V$ consists of $d$-connected components that we call the \emph{lemniscate archipelago} following \cite{ameur2021szego}, see Figure~\ref{Fig:LemAZ}.

We denote by $q_{j,N}^c$ the orthonormal polynomials associated with the weighted measure $e^{-N V_c}\, dA$:
\begin{equation} \label{OP Vc}
\int_{ \C }  q_{j,N}^c(z) \overline{ q_{k,N}^c (z) } |z|^{2c} e^{-N V(z)}\,dA(z)=\delta_{jk}.
\end{equation}
For $a>1$, it was shown in \cite{MR3668632,MR3670735} that as $j \to \infty$, the (non-trivial) zeros of $q_{j,N}^c$ tend to accumulate on the curve 
\begin{equation} \label{Sa d a>1}
\mathcal{S}_a^d := \Big \{ z \in \C: \log|z^d-a|+a\re z^d=\log \Big(\frac{1}{a}\Big)-1+a^2 \, , \, \re z^d \ge a-\frac{1}{a} \Big \}.
\end{equation}
Notice that \eqref{Sa d a>1} and \eqref{Sa a>1} are related by the mapping $z \mapsto a-z^d$.
See Figure~\ref{Figure_Sad} for the shape of $\SS_a^d$.

\begin{figure}[h!]
	\begin{subfigure}{0.32\textwidth}
		\begin{center}	
			\includegraphics[width=\textwidth]{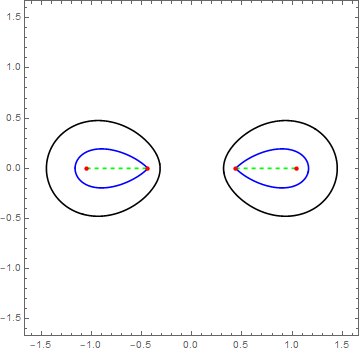}
		\end{center}
		\subcaption{$d=2$}
	\end{subfigure}	
	\begin{subfigure}[h]{0.32\textwidth}
		\begin{center}
			\includegraphics[width=\textwidth]{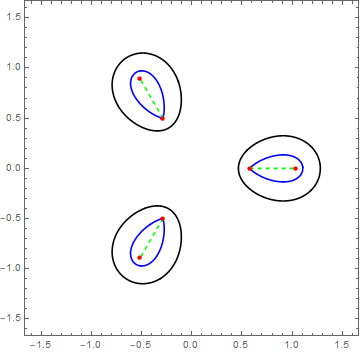}
		\end{center} \subcaption{$d=3$}
	\end{subfigure}
		\begin{subfigure}[h]{0.32\textwidth}
		\begin{center}
			\includegraphics[width=\textwidth]{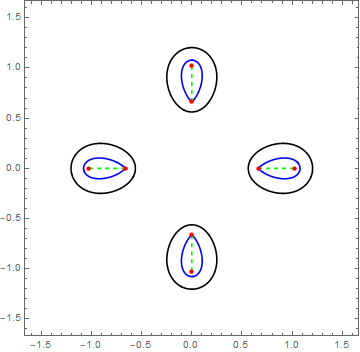}
		\end{center} \subcaption{$d=4$}
	\end{subfigure}
	\caption{ The plots display $\pa S_V$ (black) and $\mathcal{S}_a^d$ (blue), where $a=1.1$. The red dots indicate $(a-\frac{1}{a})^{1/d} \omega^k$ and $a^{1/d} \omega^k$. The green dashed lines are the branch cuts in \eqref{KN asym v2}. }  \label{Figure_Sad}
\end{figure}

Let us consider the associated correlation kernel 
\begin{equation} \label{KN Vc}
K_{N}^{c}(z,w):=|z w|^ce^{ -\frac{N}{2} (V(z)+V(w))  } \sum_{j=0}^{N-1} q_{j,N}^c(z) \overline{q_{j,N}^c(w)}.
\end{equation}
The kernel \eqref{KN Vc} corresponds to the reproducing kernel \eqref{KN ONP} associated with the potential $Q=V_c.$
We derive the asymptotic behaviours of $K_N^c$ in the macroscopic scale. 

\begin{thm} \label{Thm_archipelago} \textup{\textbf{(Macroscopic asymptotic of the lemniscate archipelago)}}
Let $d > 1$, $a>1$ and $c>-1$ be fixed. 
Suppose that $z$ and $w$ are outside $\SS_a^d$, and $|z-w| > \delta $ for some $\delta >0$. 
If $c=0,1,\dots,d-1$, we further assume that $(z^d-a)(\bar{w}^d-a)$ is outside $\SS_1$.
Then as $N \to \infty,$ we have 
\begin{align}
\begin{split} \label{KN asym v2}
	K_{dN}^c(z,w) & = d\sqrt{ \frac{N}{2\pi } } \frac{ ( (z^d-a)(\bar{w}^d-a) )^{N}  }{(z^d-a)(\bar{w}^d-a)-1}\,|z \bar{w}|^{c}     e^{ N -\frac{dN}{2} ( V(z)+V(w) ) } 
	\\
	& \quad \times \Big( \frac{z^d-a}{az^d+1-a^2} \frac{ \bar{w}^d-a }{ a\bar{w}^d+1-a^2 } \Big)^{ \frac{c}{d} }  ((z^d-a)(\bar{w}^d-a))^{\frac{1}{d}-1}
\\
& \quad \times   \frac{ (az^d+1-a^2)(a\bar{w}^d+1-a^2)-z^d(z^d-a) \bar{w}^d ( \bar{w}^d-a )  }{      (az^d+1-a^2)^{\frac{1}{d}}(a\bar{w}^d+1-a^2)^{\frac{1}{d}}-z(z^d-a)^{\frac{1}{d}}\bar{w}( \bar{w}^d-1 )^{\frac{1}{d}}    } \cdot \Big(1+O(\frac{1}{N})\Big).
\end{split}
\end{align}
Here the branch cuts for the variables $z$ and $\bar{w}$ are given by the combination of $d$ line segments connecting $(a-\frac{1}{a})^{1/d} \omega^k$ and $a^{1/d} \omega^k$, where $\omega=e^{2\pi i/d}$ and $k=0,1,\dots,d-1.$ 
\end{thm}

Note that by \eqref{V Vc} and \eqref{SV droplet}, we have 
\begin{equation}
|z^d-a|=1,  \qquad V(z)=\frac{1}{d}, \qquad (z \in \pa S_V).
\end{equation}
Then as an immediate consequence of Theorem~\ref{Thm_archipelago}, we obtain that for $z,w \in \pa S_V$, 
\begin{align}
\begin{split}
\label{KN boundary}
	|K_{dN}^c(z,w)| & = d\sqrt{ \frac{N}{2\pi } } \Big|  \frac{ (z \bar{w})^{c}  }{(z^d-a)(\bar{w}^d-a)-1}    \Big( (az^d+1-a^2)(a\bar{w}^d+1-a^2)  \Big)^{ -\frac{c}{d} }  \Big|
\\
&\quad \times \Big|  \frac{ (az^d+1-a^2)(a\bar{w}^d+1-a^2)-z^d(z^d-a) \bar{w}^d ( \bar{w}^d-a )  }{      (az^d+1-a^2)^{\frac{1}{d}}(a\bar{w}^d+1-a^2)^{\frac{1}{d}}-z(z^d-a)^{\frac{1}{d}}\bar{w}( \bar{w}^d-1 )^{\frac{1}{d}}    } \Big| \cdot \Big(1+O(\frac{1}{N})\Big). 
\end{split}
\end{align}
Thus one can notice that $K_{dN}^c(z,w)=O(\sqrt{N})$ for $z,w \in \pa S_V$, which indicates that there are strong correlations among the particles on the boundary of the droplets. 

To provide a physical realisation of Theorem~\ref{Thm_archipelago}, let us consider the Berezin kernel
\begin{equation} \label{Berezin}
B_N(z,w):=\frac{ R_{N,1}(z)R_{N,1}(w)-R_{N,2}(z,w) }{ R_{N,1}(z) } = \frac{ |K_N^c(z,w)|^2 }{ R_{N,1}(z) }.
\end{equation}
For a given point $z$, the function $w \mapsto B_N(z,w)$ corresponds to the probability density of the ensemble conditioned to have a particle at $z.$ 
See Figure~\ref{Fig_CDensity} for the graphs of $B_N$.

\begin{figure}[h!]
			\begin{subfigure}{0.32\textwidth}
		\begin{center}	
			\includegraphics[width=\textwidth]{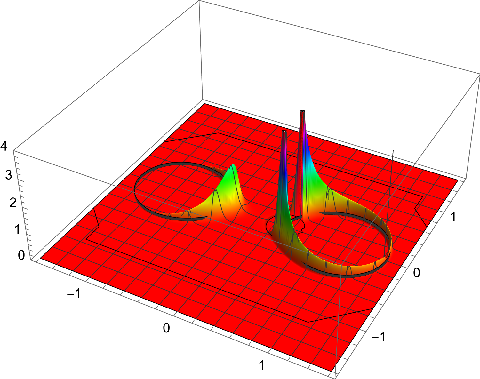}
		\end{center}
		\subcaption{$d=2$}
	\end{subfigure}	
	\begin{subfigure}[h]{0.32\textwidth}
		\begin{center}
			\includegraphics[width=\textwidth]{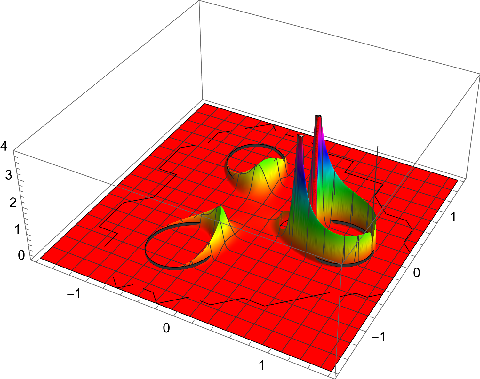}
		\end{center} \subcaption{$d=3$}
	\end{subfigure}
		\begin{subfigure}[h]{0.32\textwidth}
		\begin{center}
			\includegraphics[width=\textwidth]{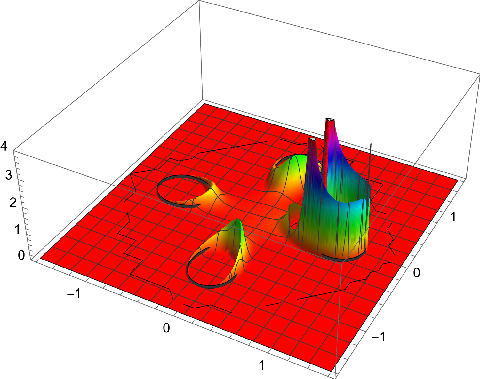}
		\end{center} \subcaption{$d=4$}
	\end{subfigure}
	\caption{ The plots display the approximation of the graphs $w\mapsto B_N(z,w)$ in \eqref{Berezin} (for $w$ away from $z$), where $z=(a-1)^{\frac{1}{d}} \in \pa S_V$ and $a=1.1$. 
	Here $c=0$ and $N=600$. 
	For the approximation, we use \eqref{KN asym v2} and the fact that $R_{N,1}(z) \sim N \Delta V(z)/2$.  } \label{Fig_CDensity} 
\end{figure}

We remark that the asymptotic behaviour \eqref{KN asym v2} may involve special functions (in the subleading terms) with certain periodicity such as Jacobi theta functions as observed in \cite{charlier2021large} for a rotationally symmetric ensemble.  
We refer to \cite{charlier2021asymptotics,MR4381929} for a discussion on similar situations on Hermitian matrix model.

In our final result, we derive the boundary scaling limits.

\begin{thm} \label{Thm_local limit} \textup{\textbf{(Boundary scaling limits of the lemniscate archipelago)}}
Let $p \in \pa S_V$ and choose $\theta$ so that $e^{i\theta}$ is outer normal to $\pa S_V$ at $p$. 
Then as $N \to \infty$, we have 
\begin{equation} \label{KN lem boundary}
\frac{1}{dN \Delta V(p)} K_{dN}^c\Big( p+ \frac{e^{i\theta}\, z}{ \sqrt{dN \Delta V(p)} } \, , \, p+ \frac{e^{i\theta}\, w}{ \sqrt{dN \Delta V(p)} }  \Big) \to  G(z,w)\, \frac12  \erfc\Big( \frac{z+\bar{w}}{\sqrt{2}} \Big),
\end{equation}
uniformly for $z,w$ on compact subsets of $\C$.
\end{thm}

We emphasise that Theorem~\ref{Thm_local limit} provides an example of edge universality for the ensembles with disconnected droplets that are not covered in \cite{hedenmalm2017planar}. 

\subsection{Outline of the proofs}

The overall strategy of the proofs is as follows.

\begin{itemize}
    \item We use the multi-fold transform of the correlation kernels (Lemma~\ref{Lem_multifold trans}) that relates those of $Q_c$ in \eqref{Q induced Gin} and of $V_c$ in \eqref{V Vc}.  
    Due to this property, one can easily derive Theorems~\ref{Thm_archipelago} and ~\ref{Thm_local limit} from Theorems~\ref{Thm_induced Ginibre} and ~\ref{Thm_Boundary iGinibre}, respectively.
    \smallskip 
    \item We apply the generalised Christoffel-Darboux formula (Proposition~\ref{Prop_CDI}) that allows expressing the correlation kernel only in terms of three monic orthogonal polynomials (of degree $N-1$, $N$, and $N+1$) and their norms. 
    \smallskip 
    \item Using the steepest descent method to the Riemann-Hilbert problem developed in \cite{MR3280250,MR3670735}, we derive the asymptotic behaviours of orthogonal polynomials (Proposition~\ref{Prop_psin diff}) and norms (Lemma~\ref{Lem_hn diff}) up to the first subleading terms. 
    Combined with the Christoffel-Darboux formula, these lead to Theorems~\ref{Thm_induced Ginibre} and ~\ref{Thm_Boundary iGinibre}.
\end{itemize}

The overall strategy described above was introduced in \cite{byun2021lemniscate} to obtain the microscopic limit of the correlation kernel at multi-criticality $a=1$. 
We use this strategy when $a>1$ together with new asymptotic behaviours of orthogonal polynomials and their norms (Proposition~\ref{Prop_psin diff} and Lemma~\ref{Lem_hn diff}). 
These are probably of interest by themselves in the spirit of several works \cite{MR3668632,lee2020strong,MR3849128,MR2921180} on Riemann-Hilbert analysis for planar orthogonal polynomials. 

The rest of this paper is organised as follows.
In Section~\ref{Section_Overall proof}, we present the overall strategy of the proofs in more detail and show our main results. 
However it requires Proposition~\ref{Prop_psin diff} and Lemma~\ref{Lem_hn diff} that are only shown in the following section.  
For the proofs, in Section~\ref{Section_fine asymptotics}, we use the nonlinear steepest descent method to the Riemann-Hilbert problem associated with the orthogonal polynomials.
In Appendix~\ref{Subsec_asymptotic c1}, we present the proofs of Proposition~\ref{Prop_psin diff} and Lemma~\ref{Lem_hn diff}  for the exactly solvable case $c=1$ using well-known properties of some special functions.

\section{Proofs of main results} \label{Section_Overall proof}

In this section, we present the overall strategy of the proofs and show the main results. 
In Subsections~\ref{Subsec_Multifold} and ~\ref{Subsec_CDI} we introduce the multi-fold transform (Lemma~\ref{Lem_multifold trans}) and the generalised Christoffel-Darboux formula. 
In Subsection~\ref{Subsec_fine asymp}, we present asymptotic behaviours of orthogonal polynomials (Proposition~\ref{Prop_psin diff}) and the norms (Lemma~\ref{Lem_hn diff}).
In Subsection~\ref{Subsec_proof thms induced}, we prove Theorems~\ref{Thm_induced Ginibre} and ~\ref{Thm_Boundary iGinibre}. 
In the last subsection, we show Theorems~\ref{Thm_archipelago} and ~\ref{Thm_local limit}.

\subsection{Multi-fold transform} \label{Subsec_Multifold}

We write $p_{j,N}^c$ for the orthonormal polynomials satisfying
\begin{equation}
\int_{ \C }  p_{j,N}^{c}(z) \overline{p_{k,N}^c (z) } |z|^{2c} e^{-N |z-a|^2}\,dA(z)=\delta_{jk}.
\end{equation}
Then the orthogonal polynomials $q_{j,N}^c$ in \eqref{OP Vc} is related to $p_{j,N}^c$ as 	
\begin{equation} \label{ONP transform}
		q_{dj+l,dN}^{c}(z)=\sqrt{d}\,z^l p_{j,N}^{ \frac{c+l+1}{d}-1 }(z^d),
\end{equation}
 see e.g. \cite[Section 3]{MR3849128} and \cite[Section 2]{byun2021lemniscate}.
We now define the correlation kernel
\begin{equation} \label{bfwhKN}
	\widehat{K}_{N}^c(z,w) =(z \bar{w})^ce^{ -\frac{N}{2} (|z-a|^2+|w-a|)^2 } \sum_{j=0}^{N-1} p_{j,N}^c(z) \overline{p_{j,N}^c(w)}. 
\end{equation}
Notice that we use $(z \bar{w})^c$ instead of $|z w|^c$. 

By \eqref{ONP transform}, we have the following multi-fold transform relation, see \cite[Section 2]{byun2021lemniscate} for more detail.  (Cf. this idea appeared also in \cite[Proposition 2.1]{claeys2008universality}, see \cite{MR4229527} for the chiral setup.)
Recall that $K_N^c$ is given by \eqref{KN Vc}.

\begin{lem} \label{Lem_multifold trans}
We have
\begin{equation}\label{multifold trans}
	K_{dN}^c(z,w) =d(z \bar{w})^{d-1} \Big( \frac{|zw|}{z\bar{w}} \Big)^c \sum_{l=0}^{d-1} \widehat{K}_N^{ \frac{c+l+1}{d}-1 }(z^d,w^d).
\end{equation}
\end{lem}
Note that in the left-hand side of \eqref{multifold trans} we use $dN$ instead of $N$. This is indeed the key observation for such a transform.  
Due to Lemma~\ref{Lem_multifold trans}, it suffices to derive the asymptotics of  $\widehat{K}_N^{c}$.

\subsection{Christoffel-Darboux formula} \label{Subsec_CDI}

One can compute asymptotics of $\bar{\partial}_w \widehat{K}_N^{c}(z,w)$ by virtue of the Christoffel-Darboux formula in \cite[Theorem 3.2]{byun2021lemniscate}. 

For this, we set some notations. 
Let $P_j \equiv P_j^c$ be the \emph{monic} orthogonal polynomial satisfying 
\begin{align} \label{Pj monic}
	\int_{ \C }  P_j(z) \overline{ P_k (z) } |z-a|^{2c} e^{-N |z|^2}\,dA(z)= h_j \, \delta_{jk},
\end{align}
where $h_j$ is the (squared) orthogonal norm. 
Note that we have the following relation 
\begin{equation}
p_{j,N}^c(z)=\frac{1}{\sqrt{h_j}} P_j(a-z). 
\end{equation}

We denote
\begin{equation} \label{W psi phi}
W(z)=(z-a)^c,\qquad	\psi_j(z):= W(z) P_j(z), \qquad \phi_j (z):=W(z) \frac{P_j(z)}{h_j}. 
\end{equation}
Let us define 
\begin{equation}
	\wt{K}_{N}^c(z,w):=((z-a)(\bar{w}-a))^c e^{-N z \bar{w} } \sum_{j=0}^{N-1} \frac{ P_j(z) \overline{ P_j(w) } }{ h_j } =e^{-N z \bar{w} } \sum_{j=0}^{N-1}  \overline{ \phi_j(w) } \psi_j(z).
\end{equation}
The kernel $K_N$ in \eqref{KN ONP} with $Q$ given by \eqref{Q induced Gin} is written in terms of $\wt{K}_{N}^c$ as
\begin{equation} \label{KN KN til}
K_N(z,w)=\Big( \frac{ |(z-a)(\bar{w}-a)| }{ (z-a)(\bar{w}-a) }\Big)^c e^{ -\frac{N}{2}(|z|^2+|w|^2-2z \bar{w}) }  \wt{K}_{N}^c(z,w). 
\end{equation}
Note also that 
\begin{equation}
	\wt{K}_{N}^c(a-z,a-w)=(z \bar{w})^c   e^{-N(a-z)(a-\bar{w})}  \sum_{j=0}^{N-1} p_{j,N}^c(z) \overline{p_{j,N}^c(w)}
\end{equation}
Thus it is related to $\wh{K}_{N}^c$ in \eqref{bfwhKN} as
\begin{equation} \label{whKN wtKN}
\wh{K}_{N}^c(z,w)=e^{ -\frac{N}{2} ( |z-a|^2+|w-a|^2-2(z-a)(\bar{w}-a) ) }  \wt{K}_{N}^c(a-z,a-w). 
\end{equation}
The following version of the Christoffel-Darboux formula was obtained in \cite[Theorem 3.2]{byun2021lemniscate}. 

\begin{prop} \textbf{\textup{(Christoffel-Darboux formula)}}
\label{Prop_CDI} 
Suppose that $a \not=0$ and that 
\begin{equation} \label{assumption_CDI}
\langle z\psi_j | \phi_0 \rangle\neq 0, \qquad \phi_j(a)\neq 0, \qquad \text{for all }j.  
\end{equation}
Then we have the following form of the Christoffel-Darboux identity:
\begin{align} \label{CDI}
	\begin{split}
		\bp_w \wt{K}_{N}^c(z,w)
		&=e^{ -N z \bar{w} }  \frac{1}{ \tfrac{N+c}{N}h_{N-1}-h_{N}  }	\bp_w \overline{ \psi_{N}(w) }    \Big( \psi_{N}(z)-z \psi_{N-1}(z) \Big)
		\\
		&\quad -e^{ -N z \bar{w} } \frac{P_{N+1}(a)}{P_N(a)} \frac{N\,h_N/h_{N-1} }{ \tfrac{N+c+1}{N} h_N-h_{N+1}   } \overline{ \psi_{N-1}(w) }  \Big(  \psi_{N+1}(z)-z \psi_N(z)  \Big).
	\end{split}
\end{align}
\end{prop}

This formula plays a key role in performing asymptotic analysis for Theorems~\ref{Thm_induced Ginibre} and ~\ref{Thm_Boundary iGinibre}.
We also refer to \cite{lee2016fine,akemann2021scaling,MR1917675,byun2021universal,byun2022almost,byun2022spherical} for various Christoffel-Darboux type identities involving certain differential operators.

\subsection{Fine asymptotic behaviours of orthogonal polynomials and norms} \label{Subsec_fine asymp}

Recall that the monic polynomial $P_j$ satisfies the orthogonality condition \eqref{Pj monic}. 
The weighted orthogonal polynomial $\psi_j$ is given by \eqref{W psi phi}.
We obtain the strong asymptotic behaviour of $\psi_j$ up to the first subleading terms.

\begin{prop} \label{Prop_psin diff}
Let $a>1$ and $c>-1$. Then for $z \in \C$ outside $\mathcal{S}_a$ in \eqref{Sa a>1}, we have
\begin{align}
 \psi_{N-1}(z)&= z^{N+c-1} \Big( \frac{z-a}{z- \frac{1}{a} } \Big)^c \cdot \Big[ 1-
\frac{c}{1-az}\Big(  \frac{1+c}{2}\frac{1}{1-az}+\frac{c}{1-a^2}-1 \Big)
\frac{1}{N}+O(\frac{1}{N^2}) \Big] , \label{psi n-1}
\\
\psi_N(z)&=z^{N+c} \Big( \frac{z-a}{z- \frac{1}{a} } \Big)^c \cdot \Big[ 1-\frac{c}{1-az}\Big(  \frac{1+c}{2}\frac{1}{1-az}+\frac{c}{1-a^2}  \Big)
\frac{1}{N}+O(\frac{1}{N^2}) \Big] ,  \label{psi n}
\\
\psi_{N+1}(z)  &= z^{N+c+1} \Big( \frac{ z-a }{ z-\frac{1}{a} } \Big)^c \cdot \Big[ 1-
\frac{c}{1-az}\Big(  \frac{1+c}{2}\frac{1}{1-az}+\frac{c}{1-a^2}+1 \Big)
\frac{1}{N}+O(\frac{1}{N^2}) \Big]. \label{psi n+1}
\end{align}
In particular, we have
\begin{align} 
\psi_N(z)-z \psi_{N-1}(z)&=z^{N+c} \Big( \frac{z-a}{z- \frac{1}{a} } \Big)^c \cdot  \frac{c}{az-1}   \frac{1}{N} \cdot \Big(1+O(\frac{1}{N}) \Big), \label{psi n n-1}
\\
\psi_{N+1}(z)-z \psi_{N}(z)&=z^{N+1+c} \Big( \frac{z-a}{z- \frac{1}{a} } \Big)^c \cdot   \frac{c}{az-1}   \frac{1}{N} \cdot \Big(1+O(\frac{1}{N}) \Big).  \label{psi n+1 n}
\end{align}
\end{prop}

We emphasise that the leading terms in Proposition~\ref{Prop_psin diff} were obtained in \cite[Theorem 2]{MR3670735}.
Note that the terms \eqref{psi n n-1} and \eqref{psi n+1 n} appear in the Christoffel-Darboux formula \eqref{CDI}.
For these terms, we should extend \cite[Theorem 2]{MR3670735} up to the first subleading $O(1/N)$ terms. 

Notice that if $c=0$, then $\psi_j(z)=z^{j}.$
Thus in this case Proposition~\ref{Prop_psin diff} trivially holds.

To apply the Christoffel-Darboux formula \eqref{CDI}, one should also derive the asymptotic behaviours of the orthogonal norms $h_j$ in \eqref{Pj monic}.

\begin{lem} \label{Lem_hn diff}
Let $a>1$ and $c>-1$. Then we have 
\begin{align}\label{hN-1 c gen}
h_{N-1}&= e^{-N} \sqrt{ \frac{2\pi}{N} } a^{2c} \cdot \Big[ 1+ \Big( \frac{c}{a^2-1}+\frac{c(c-1)}{2}+\frac{1}{12}\Big)\frac{1}{N}+O(\frac{1}{N^2}) \Big] ,
\\\label{hN c gen}
h_N &= e^{-N} \sqrt{ \frac{2\pi}{N} } a^{2c} \cdot \Big[ 1+ \Big( \frac{c}{a^2-1}+\frac{c(c-1)}{2}+\frac{1}{12}\Big)\frac{1}{N}+O(\frac{1}{N^2}) \Big] ,
\\\label{hN+1 c gen}
h_{N+1} &= e^{-N} \sqrt{ \frac{2\pi}{N} } a^{2c} \cdot \Big[ 1+ \Big( \frac{c}{a^2-1}+\frac{c(c-1)}{2}+\frac{13}{12}\Big)\frac{1}{N}+O(\frac{1}{N^2}) \Big].
\end{align}
In particular, for $c \not =0$, we have
\begin{align}\label{eq hNN-1}
\frac{1}{ \tfrac{N+c}{N}h_{N-1}-h_{N}  }& = \frac{1}{c} \frac{1}{a^{2c} }\frac{1}{\sqrt{2\pi}}\, N^{\frac32}\, e^N \cdot  \Big( 1+O(\frac{1}{N}) \Big), 
 \\
 \label{eq hNN-1N+1}
 \frac{N\,h_N/h_{N-1} }{ \tfrac{N+c+1}{N} h_N-h_{N+1}   } &=  \frac{1}{c}  \frac{1}{a^{2c} }\frac{1}{\sqrt{2\pi}}\, N^{\frac52}\, e^N \cdot  \Big( 1+O(\frac{1}{N}) \Big). 
\end{align}
\end{lem}

Note that \eqref{eq hNN-1} and \eqref{eq hNN-1N+1} appear in the Christoffel-Darboux formula \eqref{CDI}.
We remark that for $c=0$, we have 
$$
h_j=\int_\C |z|^{2j} e^{-N |z|^2}\,dA(z)=2\int_0^\infty r^{2j+1} e^{-Nr^2}\,dr= \frac{j!}{N^{j+1}}.
$$
Thus by Stirling's formula, one can directly check that Lemma~\ref{Lem_hn diff} holds for $c=0.$

\subsection{Proofs of Theorems~\ref{Thm_induced Ginibre} and ~\ref{Thm_Boundary iGinibre} } \label{Subsec_proof thms induced}

Combining the Christoffel-Darboux formula (Proposition~\ref{Prop_CDI}) with Proposition~\ref{Prop_psin diff} and Lemma~\ref{Lem_hn diff}, we show Theorem~\ref{Thm_induced Ginibre}.

\begin{proof}[Proof of Theorem~\ref{Thm_induced Ginibre}]
By the transform \eqref{KN KN til}, it suffices to show that for $z$ and $\bar{w}$ outside $\SS_a$ in \eqref{Sa a>1}, 
\begin{equation} \label{wt K asymp}
 \wt{K}_{N}^c(z,w)=  \sqrt{ \frac{N}{2\pi} } \frac{1}{z\bar{w}-1}\,\Big( \frac{z-a}{1-az}  \frac{ \bar{w} -a}{ 1-a\bar{w} } \Big)^c \cdot (z\bar{w})^{N+c}   e^{ N-N z \bar{w} } \cdot  \Big(1+O(\frac{1}{N})\Big).
\end{equation}
Using Proposition~\ref{Prop_psin diff}, we have
\begin{equation}
\bp_w \overline{ \psi_{N}(w) } = N \bar{w}^{N+c-1} \Big( \frac{ \bar{w}-a }{ \bar{w}-\frac{1}{a} } \Big)^c \cdot \Big(1+O(\frac{1}{N}) \Big), \qquad 
 \overline{ \psi_{N-1}(w) }  = \bar{w}^{N+c-1} \Big( \frac{ \bar{w}-a }{ \bar{w}-\frac{1}{a} } \Big)^c \cdot \Big(1+O(\frac{1}{N}) \Big). 
\end{equation}
By \cite[Theorem 2]{MR3670735}, we also have 
\begin{equation}
\frac{P_{N+1}(a)}{P_N(a)}= a\cdot \Big( 1+O(\frac{1}{N}) \Big).
\end{equation}
Therefore by \eqref{psi n n-1} and \eqref{psi n+1 n}, we obtain
\begin{equation}
\bp_w \overline{ \psi_{N}(w) }    \Big( \psi_{N}(z)-z \psi_{N-1}(z) \Big)=  \frac{c}{az-1}    z^{N+c}   \bar{w}^{N+c-1} \Big( \frac{z-a}{z- \frac{1}{a} } \Big)^c  \Big( \frac{ \bar{w}-a }{ \bar{w}-\frac{1}{a} } \Big)^c \cdot \Big(1+O(\frac{1}{N}) \Big) 
\end{equation}
and 
\begin{equation}
\begin{split}
\overline{ \psi_{N-1}(w) }  \Big(  \psi_{N+1}(z)-z \psi_N(z)  \Big) 
=  \frac{c}{az-1}  \frac{1}{N} z^{N+1+c}  \bar{w}^{N+c-1} \Big( \frac{z-a}{z- \frac{1}{a} } \Big)^c \Big( \frac{ \bar{w}-a }{ \bar{w}-\frac{1}{a} } \Big)^c \cdot \Big(1+O(\frac{1}{N}) \Big).
\end{split}
\end{equation}
Then by Lemma~\ref{Lem_hn diff}, we have
\begin{equation}
\begin{split}
&\quad  \frac{1}{ \tfrac{N+c}{N}h_{N-1}-h_{N}  }	\bp_w \overline{ \psi_{N}(w) }    \Big( \psi_{N}(z)-z \psi_{N-1}(z) \Big)
\\
&=  \frac{1}{c} \frac{1}{a^{2c} }\frac{1}{\sqrt{2\pi}}\, N^{\frac32}\, e^N         \frac{c}{az-1} z^{N+c}   \bar{w}^{N+c-1} \Big( \frac{z-a}{z- \frac{1}{a} } \Big)^c  \Big( \frac{ \bar{w}-a }{ \bar{w}-\frac{1}{a} } \Big)^c \cdot \Big(1+O(\frac{1}{N}) \Big) 
\\
&=  \frac{1}{\sqrt{2\pi}}\, N^{\frac32}\, e^N     \frac{1}{az-1}  z^{N+c}   \bar{w}^{N+c-1}  \Big( \frac{z-a}{1-az}  \frac{ \bar{w} -a}{ 1-a\bar{w} } \Big)^c \cdot \Big(1+O(\frac{1}{N}) \Big) 
\end{split}
\end{equation}
and
\begin{equation}
\begin{split}
&\quad \frac{P_{N+1}(a)}{P_N(a)} \frac{N\,h_N/h_{N-1} }{ \tfrac{N+c+1}{N} h_N-h_{N+1}   } \overline{ \psi_{N-1}(w) }  \Big(  \psi_{N+1}(z)-z \psi_N(z)  \Big)
\\
&= a\,  \frac{1}{c}  \frac{1}{a^{2c} }\frac{1}{\sqrt{2\pi}}\, N^{\frac52}\, e^N       \frac{c}{az-1}  \frac{1}{N} z^{N+1+c}  \bar{w}^{N+c-1} \Big( \frac{z-a}{z- \frac{1}{a} } \Big)^c \Big( \frac{ \bar{w}-a }{ \bar{w}-\frac{1}{a} } \Big)^c \cdot \Big(1+O(\frac{1}{N}) \Big)
\\
&=    \frac{1}{\sqrt{2\pi}}\, N^{\frac32}\, e^N   \frac{az}{az-1}   z^{N+c}  \bar{w}^{N+c-1} \Big( \frac{z-a}{1-az}  \frac{ \bar{w} -a}{ 1-a\bar{w} } \Big)^c  \cdot \Big(1+O(\frac{1}{N}) \Big).
\end{split}
\end{equation}
Now it follows from the Christoffel-Darboux formula (Proposition~\ref{Prop_CDI}) that  
\begin{equation}
	\bp_w \wt{K}_{N}^{c}(z,w)= -\frac{N^{\frac32}}{\sqrt{2\pi}}\, \Big( \frac{z-a}{1-az}  \frac{ \bar{w} -a}{ 1-a\bar{w} } \Big)^c  \cdot z^{N+c}  \bar{w}^{N+c-1}  e^{ N-N z \bar{w} } \cdot  \Big(1+O(\frac{1}{N})\Big). 
\end{equation}
Integrating this equation, we obtain 
\begin{equation}
\wt{K}_{N}^c(z,w)=  \sqrt{ \frac{N}{2\pi} } \frac{1}{z\bar{w}-1}\,\Big( \frac{z-a}{1-az}  \frac{ \bar{w} -a}{ 1-a\bar{w} } \Big)^c \cdot (z\bar{w})^{N+c}   e^{ N-N z \bar{w} } \cdot  \Big(1+O(\frac{1}{N})\Big)+ f_N(z)
\end{equation}
for some function $f_N$ depending only on $z$. 
Due to the symmetry $\wt{K}_N^c(z,w)=\wt{K}_N^c(w,z)$, it follows that $f_N$ is a constant function. 
Furthermore, by combining the exterior estimate
$$
\wt{K}_N^c(z,z) \to 0, \qquad \mbox{as } z \to \infty  
$$
that holds in general (see \cite[Section 4.1.1]{AHM15}) and the elementary inequality
$$
\det \begin{bmatrix}
\wt{K}_N^c(z,z) & \wt{K}_N^c(z,w)
\\
\wt{K}_N^c(w,z) & \wt{K}_N^c(w,w)
\end{bmatrix} \ge 0, \qquad \mbox{i.e.} \qquad  \wt{K}_N^c(z,z) \wt{K}_N^c(w,w) \ge | \wt{K}_N^c(z,w) |^2, 
$$
one can observe that $\wt{K}_N^c(z,w) \to 0$ as $z \to \infty$. 
Thus we conclude \eqref{wt K asymp}. 
\end{proof}

\begin{proof}[Proof of Theorem~\ref{Thm_Boundary iGinibre}]
By \eqref{KN KN til}, it suffices to show that 
\begin{equation} \label{cov of wt K local}
\frac{1}{N}\wt{K}_N^c\Big(p+\frac{p\,z}{\sqrt{N}}\,,\, p+\frac{p\, w}{\sqrt{N}}\Big) \to  \frac12 \erfc\Big( \frac{z+\bar{w}}{\sqrt{2}} \Big).
\end{equation}
To lighten notations, let us write 
\begin{equation}
z_p:= p\,\Big(1+\frac{z}{\sqrt{N}}\Big), \qquad w_p:= p\,\Big(1+\frac{w}{\sqrt{N}}\Big).
\end{equation}
First note that 
\begin{equation}
e^{-N z_p \bar{w}_p}= e^{-N-\sqrt{N}(z+\bar{w}) -z\bar{w} }.
\end{equation}
We also have
\begin{equation}
z_p^{N+c}= p^{N+c}\,e^{ \sqrt{N}z -\frac{z^2}{2} }\cdot \Big( 1+O(\frac{1}{\sqrt{N}}) \Big), \qquad 
z_p^{N+1+c}= p^{N+1+c}\,e^{ \sqrt{N}z -\frac{z^2}{2} }\cdot \Big( 1+O(\frac{1}{\sqrt{N}}) \Big).
\end{equation}
Combining these asymptotics with Proposition~\ref{Prop_psin diff}, we have
\begin{align}
\psi_N(z_p)-z_p \psi_{N-1}(z_p)&=p^{N+c}\,e^{ \sqrt{N}z -\frac{z^2}{2} } \Big( \frac{p-a}{p- \frac{1}{a} } \Big)^c \cdot  \frac{c}{ap-1}   \frac{1}{N} \cdot \Big( 1+O(\frac{1}{\sqrt{N}}) \Big),
\\
\psi_{N+1}(z_p)-z_p \psi_{N}(z_p)&=p^{N+1+c}\,e^{ \sqrt{N}z -\frac{z^2}{2} } \Big( \frac{p-a}{p- \frac{1}{a} } \Big)^c \cdot   \frac{c}{ap-1}   \frac{1}{N} \cdot \Big( 1+O(\frac{1}{\sqrt{N}}) \Big)
\end{align}
and
\begin{align}
\overline{ \psi'_{N}(w_p) } & = N \, \bar{p}^{N+c-1} e^{ \sqrt{N}\bar{w}-\frac{\bar{w}^2}{2} } \Big( \frac{ \bar{p}-a }{ \bar{p}-\frac{1}{a} } \Big)^c \cdot \Big(1+O(\frac{1}{\sqrt{N}}) \Big),
\\
 \overline{ \psi_{N-1}(w_p) } & = \bar{p}^{N+c-1} e^{ \sqrt{N}\bar{w}-\frac{\bar{w}^2}{2} }  \Big( \frac{ \bar{p}-a }{ \bar{p}-\frac{1}{a} } \Big)^c \cdot \Big(1+O(\frac{1}{\sqrt{N}}) \Big). 
\end{align}
Then by Lemma~\ref{Lem_hn diff}, we obtain 
\begin{equation}
\begin{split}
&\quad e^{ -N z_p \bar{w}_p }  \frac{1}{ \tfrac{N+c}{N}h_{N-1}-h_{N}  }	 \overline{ \psi'_{N}(w_p) }    \Big( \psi_{N}(z_p)-z_p \psi_{N-1}(z_p) \Big)
\\
&= e^{ -\frac{(z+\bar{w})^2}{2} }   \frac{1}{a^{2c} }\frac{1}{\sqrt{2\pi}}\, N^{\frac32} \,   \Big( \frac{ \bar{p}-a }{ \bar{p}-\frac{1}{a} } \Big)^c \, \Big( \frac{p-a}{p- \frac{1}{a} } \Big)^c \cdot  \frac{p}{ap-1}  \cdot \Big( 1+O(\frac{1}{\sqrt{N}}) \Big)
\\
&= e^{ -\frac{(z+\bar{w})^2}{2} }  \frac{1}{\sqrt{2\pi}}\, N^{\frac32} \, \cdot  \frac{p}{ap-1}  \cdot \Big( 1+O(\frac{1}{\sqrt{N}}) \Big).
\end{split}
\end{equation}
Here we have used that
\begin{equation}
\Big| \frac{p-a}{ap- 1 } \Big|= \Big| \frac{p-a}{a- \bar{p} } \Big|=1, 
\end{equation}
which follows from $|p|=1$.

Similarly, we have 
\begin{equation}
\begin{split}
&\quad e^{ -N z_p \bar{w}_p } \frac{P_{N+1}(a)}{P_N(a)} \frac{N\,h_N/h_{N-1} }{ \tfrac{N+c+1}{N} h_N-h_{N+1}   } \overline{ \psi_{N-1}(w_p) }  \Big(  \psi_{N+1}(z_p)-z_p \psi_N(z_p)  \Big)
\\
&= e^{ -\frac{(z+\bar{w})^2}{2} }\,  \frac{1}{a^{2c} }\frac{1}{\sqrt{2\pi}}\, N^{\frac32}\,    \Big( \frac{ \bar{p}-a }{ \bar{p}-\frac{1}{a} } \Big)^c \,  \Big( \frac{p-a}{p- \frac{1}{a} } \Big)^c \cdot   \frac{ap^2}{ap-1}    \cdot \Big( 1+O(\frac{1}{\sqrt{N}}) \Big)
\\
&= e^{ -\frac{(z+\bar{w})^2}{2} }\, \frac{1}{\sqrt{2\pi}}\, N^{\frac32}\,    \cdot   \frac{ap^2}{ap-1}    \cdot \Big( 1+O(\frac{1}{\sqrt{N}}) \Big).
\end{split}
\end{equation}
Therefore by Proposition~\ref{Prop_CDI}, we obtain that 
\begin{equation}
\bp_w \Big[ \frac{1}{N}\wt{K}_N^c(z_p,w_p) \Big]=- \frac{1}{\sqrt{2\pi}} e^{-\frac{(z+\bar{w}^2)}{2} }\cdot \Big( 1+O(\frac{1}{\sqrt{N}}) \Big). 
\end{equation}
This gives that 
\begin{equation}
\frac{1}{N}\wt{K}_N^c(z_p,w_p)= \frac12 \erfc\Big( \frac{z+\bar{w}}{\sqrt{2}} \Big)\cdot  \Big( 1+O(\frac{1}{\sqrt{N}}) \Big),
\end{equation}
which leads to the desired convergence \eqref{cov of wt K local}. 
Here the integration constant is determined similarly above by the fact that $\frac{1}{N}\wt{K}_N^c(z_p,w_p) \to 0$ as $z \to +\infty$.
\end{proof}

\subsection{Proofs of Theorems~\ref{Thm_archipelago} and ~\ref{Thm_local limit} } \label{Subsec_proof theorems archipelago}

We derive Theorem~\ref{Thm_archipelago} from Theorem~\ref{Thm_induced Ginibre}.

\begin{proof}[Proof of Theorem~\ref{Thm_archipelago}]
By Theorem~\ref{Thm_induced Ginibre} and \eqref{whKN wtKN}, we have
\begin{equation}
\begin{split}
    \wh{K}_{N}^c(z,w)&=  \sqrt{ \frac{N}{2\pi} } \frac{1}{(a-z)(a-\bar{w})-1}\,\Big( \frac{z}{a^2-1-az} \frac{ \bar{w} }{ a^2-1-a\bar{w} } \Big)^c  
    \\
    &\quad \times \Big( (z-a)(\bar{w}-a)\Big)^{N+c}   e^{ N -\frac{N}{2} ( |z-a|^2+|w-a|^2 ) } \cdot \Big(1+O(\frac{1}{N})\Big). 
\end{split}
\end{equation}
By \eqref{multifold trans}, we obtain 
\begin{equation}
\begin{split} \label{KN asym v1}
	K_{dN}^c(z,w) & =d \sqrt{ \frac{N}{2\pi } } \frac{ ( (z^d-a)(\bar{w}^d-a) )^{ N }  }{(z^d-a)(\bar{w}^d-a)-1}\,|z w|^{c}     e^{ N -\frac{dN}{2} ( V(z)+V(w) ) }  \\
	&\quad \times \sum_{l=0}^{d-1} \Big( \frac{z^d-a}{az^d+1-a^2} \frac{ \bar{w}^d-a }{ a\bar{w}^d+1-a^2 } \Big)^{ \frac{c+l+1}{d}-1 }  (z\bar{w})^l  \cdot \Big(1+O(\frac{1}{N})\Big). 
\end{split}
\end{equation}
Now \eqref{KN asym v2} follows from straightforward computations. 

\end{proof}

Finally, we derive Theorem~\ref{Thm_local limit} from Theorem~\ref{Thm_Boundary iGinibre}. 

\begin{proof}[Proof of Theorem~\ref{Thm_local limit}]
By \eqref{multifold trans} and \eqref{density lemniscate}, we have
\begin{equation}
\begin{split}
&\quad \frac{1}{dN \Delta V(p)} K_{dN}^c\Big( p+ \frac{e^{i\theta}\, z}{ \sqrt{dN \Delta V(p)} } \, , \, p+ \frac{e^{i\theta}\, w}{ \sqrt{dN \Delta V(p)} }  \Big) 
\\
&=\frac{1}{d^2 N p^{2d-2} } K_{dN}^c\Big( p+ \frac{e^{i\theta}\, z}{ d\, p^{d-1}\sqrt{N} } \, , \, p+ \frac{e^{i\theta}\, w}{ d\, p^{d-1}\sqrt{N} }  \Big) 
\\
&= \frac{1}{dN } \sum_{l=0}^{d-1} \widehat{K}_N^{ \frac{c+l+1}{d}-1 }\Big(p^d+\frac{e^{i\theta}\,z }{ \sqrt{N} }\,, \, p^d+\frac{e^{i\theta}\,w }{ \sqrt{N} } \Big) \cdot \Big( 1+O(\frac{1}{\sqrt{N}}) \Big).
\end{split}
\end{equation}
Therefore Theorem~\ref{Thm_local limit} follows from Theorem~\ref{Thm_Boundary iGinibre}.
\end{proof}

\section{Riemann-Hilbert analysis and fine asymptotic behaviours} \label{Section_fine asymptotics}

In this section, we derive fine asymptotic behaviours of the orthogonal polynomials (Proposition~\ref{Prop_psin diff}) and the orthogonal norms (Lemma~\ref{Lem_hn diff}). 
Subsection~\ref{Subsec_outline RH} is devoted to the recalling the matrix-valued Riemann-Hilbert problem developed in \cite{MR3280250} and the transforms introduced in \cite{MR3670735}. 
Based on the Riemann-Hilbert analysis in Subsections~\ref{Subsec_asymptotic OPs} and ~\ref{Subsec_asymptotic norms}, we prove Proposition~\ref{Prop_psin diff} and Lemma~\ref{Lem_hn diff}.

\subsection{Outline of the Riemann-Hilbert analysis} \label{Subsec_outline RH}

Let us briefly recall the Riemann-Hilbert analysis in \cite{MR3670735,MR3280250} (see also \cite{MR3962350,lee2020strong} for its generalisation) that was developed to derive the asymptotic behaviours of the orthogonal polynomials $P_n.$ 
We also refer the reader to \cite{MR3849128,MR3668632} for similar studies in different settings. 
This will be used in the following subsection to derive fine asymptotic behaviours of $P_n$.

Let $\Gamma$ be a simple closed curve that encloses the line segment $[0,a]\in \C$ with counterclockwise orientation. Let the analytic function  $w_{n,N}$ on $\C\setminus [0,a]$ be defined by
\begin{equation} \label{weight contour}
w_{n,N}(z):=\Big(\frac{z-a}{z}\Big)^c\frac{e^{-Naz}}{z^n},
\end{equation}
where we choose the principal branch.

Define the matrix function $Y(z)$ by
\begin{equation}
Y(z):= 
\begin{bmatrix}
P_n(z)&\displaystyle\frac{1}{2\pi i}\int_\Gamma\frac{P_n(s)w_{n,N}(s)}{s-z}\, ds
\smallskip 
\\
   Q_{n-1}(z)&\displaystyle\frac{1}{2\pi i}\int_\Gamma\frac{Q_{n-1}(s)w_{n,N}(s)}{s-z} \, ds
\end{bmatrix},
\end{equation}
where $Q_{n-1}$ is a unique polynomial of degree $n-1$ satisfying 
$$
\displaystyle\frac{1}{2\pi i}\int_\Gamma\frac{Q_{n-1}(s)w_{n,N}(s)}{s-z} \, ds =\frac{1}{z^n} \cdot \Big( 1+O(\frac{1}{z}) \Big).
$$
Then it was shown in \cite[Section 3]{MR3280250} that $Y(z)$ is a unique solution to the Riemann-Hilbert problem
\begin{equation} \label{RHP Y}
    \begin{cases}
    Y(z) \mbox{ is holomorphic in $\C\setminus\Gamma$},
    \smallskip 
    \\
    Y_+(z)=Y_-(z)\begin{bmatrix}1&w_{n,N}(z) \\ 
    0&1\end{bmatrix},\quad &z\in \Gamma,
    \smallskip 
    \\
     Y(z)=\Big(I+O(\frac{1}{N})\Big)\begin{bmatrix}z^n&0\\
    0&z^{-n}\end{bmatrix},\quad &z\to \infty.
    \end{cases}
\end{equation}
Here $Y_\pm(z)$ are the boundary values on the sides of the corresponding contour.
Since $P_n(z)=[Y(z)]_{11}$, we aim to analyse the solution to the Riemann-Hilbert problem \eqref{RHP Y}. 
For this purpose, we shall introduce several transforms of \eqref{RHP Y}. 

First, let us define $g$ by
\begin{equation}
g(z)=\begin{cases}
\log z, &z\in \overline{{\rm Ext} {\mathcal S}_a},
\smallskip 
\\
az+\log\beta-a\beta, &z\in {\rm Int} {\mathcal S}_a.
\end{cases}
\end{equation}
Here and in the sequel, we write $\beta=1/a$.
The function $g$ is a building block to define
\begin{equation} \label{phi robin}
\phi(z)=az+\log z-2g(z)+l,\qquad l=\log\beta-a\beta, 
\end{equation}
which satisfies $\re \phi(z)=0$ for $z\in {\mathcal S}_a$.

Following the nonlinear steepest descent method that applied to the above Riemann-Hilbert problem for $Y$, we define 
\begin{equation}\label{rhp z}
Z(z):=e^{\frac{-Nl}{2}\sigma_3}Y(z)e^{-Ng(z)\sigma_3}e^{\frac{Nl}{2}\sigma_3}\begin{bmatrix}1&0\\
    \star \big(\frac{z}{z-a}\big)^ce^{N\phi(z)}&1\end{bmatrix},    \end{equation}
    where
    $$
    \star=\begin{cases}
    1,&\mbox{when $z\in U\cap{\rm Ext}\Gamma$},\\
    -1,&\mbox{when $z\in U\cap{\rm Int}\Gamma$},\\
    0,&\mbox{when $z\notin U$}.
    \end{cases}
    $$
Here $U$ is a neighbourhood of $\SS_a$.
Then by \eqref{RHP Y}, the matrix function $Z$ satisfies the following Riemann-Hilbert problem
\begin{equation}\label{RHP Z}
    \begin{cases}
    Z_+(z)=Z_-(z)\begin{bmatrix}1&0\\
   \big(\frac{z}{z-a}\big)^ce^{N\phi(z)}&1\end{bmatrix},\quad &z\in \partial U,
   \smallskip 
   \\
     Z_+(z)=Z_-(z)\begin{bmatrix}0&\big(\frac{z-a}{z}\big)^c\\
   -\big(\frac{z}{z-a}\big)^c&0\end{bmatrix},\quad &z\in \Gamma\cap U,
   \smallskip 
   \\  
   Z_+(z)=Z_-(z)\begin{bmatrix}1&\big(\frac{z-a}{z}\big)^ce^{-N\phi(z)}\\
   0&1\end{bmatrix},\quad &z\in \Gamma\setminus U,
   \smallskip 
   \\
     Z(z)=I+O(\frac{1}{N}),\quad &z\to \infty,
     \smallskip 
     \\
     Z(z) \mbox{ is holomorphic}, &\mbox{otherwise}.\\
    \end{cases}
\end{equation}

Next, we define the global parametrix
\begin{equation}
  \Phi(z)=  \begin{cases}
    \begin{bmatrix}\big(\frac{z}{z-\beta}\big)^c&0\\
  0& \big(\frac{z-\beta}{z}\big)^c\end{bmatrix},\quad &z\in {\rm Ext}\Gamma,
  \smallskip 
  \\
   \begin{bmatrix}0&\big(\frac{z-a}{z-\beta}\big)^c\\
   -\big(\frac{z-\beta}{z-a}\big)^c&0\end{bmatrix},\quad &z\in {\rm Int}\Gamma.
    \end{cases}
\end{equation}
Then $\Phi$ satisfies the following Riemann-Hilbert problem
\begin{equation}
    \begin{cases}
     \Phi_+(z)=\Phi_-(z)\begin{bmatrix}0&\big(\frac{z-a}{z}\big)^c\\
   -\big(\frac{z}{z-a}\big)^c&0\end{bmatrix},\quad &z\in {\mathcal S}_a,
   \smallskip 
   \\ 
     \Phi(z)=I+O(\frac{1}{N}),\quad &z\to \infty,
     \smallskip 
     \\
     \Phi(z) \quad \mbox{is holomorphic}, &\mbox{otherwise}.\\
    \end{cases}
\end{equation} 
Note that we let $\Gamma$ match $\SS_a$ for $z\in U$ and away from a small neighborhood of $\beta$.

Near the point $\beta$, the jump matrices of $\Phi$ do not converge to those of $Z$.  
Therefore one needs the local parametrix around $\beta$ that satisfies the exact jump conditions of $Z$.  
Moreover, we shall construct a rational matrix function $R$ such that the improved global parametrix, $R\Phi$, matches the local parametrix better. 
This construction is called ``partial  Schlesinger transform'' \cite{bertola2009first}, and it was used in \cite{MR3280250} to obtain the strong asymptotics of $P_n$. Here we use it to derive fine asymptotic behaviours of the orthogonal polynomials (Proposition~\ref{Prop_psin diff}) and the orthogonal norms (Lemma~\ref{Lem_hn diff}).  

Let $D_\beta$ be a disk neighborhood of $\beta$ with a fixed radius such that the map $\zeta: D_\beta\to\C$ given by
\begin{equation}
\zeta:=\sqrt{2N(a(z-\beta)-\log z+\log\beta)}=a\sqrt N(z-\beta)(1+O(z-\beta))
\end{equation}
is univalent.

\begin{figure}[h!]  
\centering
\begin{tikzpicture}[scale=0.8]
\filldraw[color=lightgray!30, fill=lightgray!30, very thin](0,0.2) rectangle (-4,4.2);
\filldraw[color=lightgray!30, fill=lightgray!30, very thin](0,-0.2) rectangle (-4,-4.2);

\draw[thin ] (0,-4.2) -- (0,4.2);
\draw[thin] (-4,0) -- (4,0);
\draw[very thick,blue, postaction={decorate, decoration={markings,  mark = at position 0.55 with {\arrow{>}}}} ] (4,0.2) -- (0,0.2);
\draw[thin,green, postaction={decorate, decoration={markings, mark = at position 0.5 with {\arrow{>}}}} ] (0,0.2) -- (-4,0.2);
\draw[thin,green, postaction={decorate, decoration={markings, mark = at position 0.5 with {\arrow{>}}}} ] (0,0.2) -- (0,4.2);
\draw[very thick,blue,postaction={decorate, decoration={markings, mark = at position 0.5 with {\arrow{>}}}} ] (0,-0.2) -- (4,-0.2);
\draw[thin,green, postaction={decorate, decoration={markings, mark = at position 0.53 with {\arrow{>}}}} ] (-4,-0.2) -- (0,-0.2);
\draw[thin,green, postaction={decorate, decoration={markings, mark = at position 0.5 with {\arrow{>}}}} ] (0,-4.2) -- (0,-0.2);
\draw[very thick,blue,postaction={decorate, decoration={markings, mark = at position 0.5 with {\arrow{>}}}} ] (0,0.2) -- (-4,4.2);
\draw[very thick,blue,postaction={decorate, decoration={markings, mark = at position 0.5 with {\arrow{>}}}} ] (-4,-4.2) -- (0,-0.2);

\foreach \Point/\PointLabel in {(2.8,1)/\Gamma,(2.8,-0.5)/\Gamma, (-1.5,2.6)/U,(-1.5,-2)/U}
\draw[fill=black]  
 \Point node[below right] {$\PointLabel$};
 \end{tikzpicture}
 \caption{The jump contours of $P(z)$ in $D_\beta$. $\Gamma$ are the blue curves, $U$ are the shaded region bounded by the green curves. }
 \end{figure}
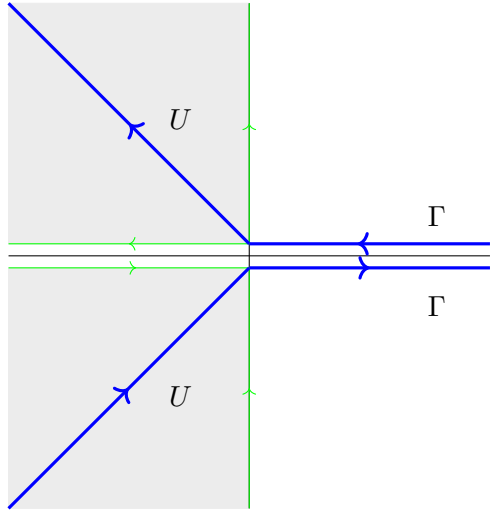

We now define $P:D_\beta\to \C^{2\times 2}$ that satisfies the following Riemann-Hilbert problem
\begin{equation} \label{RHP P}
    \begin{cases}
    P_+(z)=P_-(z)\begin{bmatrix}1&e^{-\frac{\zeta(z)^2}{2}}\\
   0&1\end{bmatrix},\quad &z\in \Gamma\setminus U,
   \smallskip 
   \\
      P_+(z)=P_-(z)\begin{bmatrix}1&0\\
  e^{\frac{\zeta(z)^2}{2}}&1\end{bmatrix},\quad &z\in \partial U\cap{\rm Ext}\Gamma,
  \smallskip 
  \\  P_+(z)=P_-(z)\begin{bmatrix}1&0\\
   e^{-\frac{\zeta(z)^2}{2}}&1\end{bmatrix},\quad &z\in \partial U\cap{\rm Int}\Gamma,
   \smallskip 
   \\
    P_+(z)=\begin{bmatrix}0&-1\\
   1&0\end{bmatrix}P_-(z)\begin{bmatrix}0&1\\
   -1&0\end{bmatrix},\quad &z\in \Gamma\cap U,
   \smallskip 
   \\
     P_+(z)=e^{-c\pi i\sigma_3}P_-(z)e^{c\pi i\sigma_3},\quad &z\in \R,
     \smallskip 
     \\
     P(z) \quad \mbox{is holomorphic}, &\mbox{otherwise}.\\
    \end{cases}
\end{equation} 
and the boundary condition, $P(z)\sim I$ on $\partial D_\beta$.
Using the Riemann-Hilbert problem \eqref{RHP P} for $P$, one can notice that the matrix function
\begin{equation}
\Phi(z)\Big(\frac{z-a}{z}\Big)^{\frac{c}{2}\sigma_3}P(z)\Big(\frac{z-a}{z}\Big)^{-\frac{c}{2}\sigma_3}
\end{equation}
satisfies the jump conditions of $Z$ in \eqref{RHP Z}.

Finally, let us define $W$ by
\begin{equation}
W(z):=\zeta(z)^{-c\sigma_3}S P(z)T(\zeta(z))^{-1}S^{-1},
\end{equation}
where $T$ is a diagonal matrix function
\begin{equation}
T(\zeta)=\begin{cases}
\exp\big(\frac{\zeta^2}{4}\sigma_3\big),\quad &|{\arg}\zeta|<3\pi/4,
\smallskip 
\\
\exp\big(-\frac{\zeta^2}{4}\sigma_3\big),\quad &\mbox{otherwise},
\end{cases}
\end{equation}
and $S$ is a piecewise constant matrix
\begin{equation}
S=\begin{cases}
I,&{\rm Im}\zeta<0\cap|{\arg}\zeta|<3\pi/4,
\smallskip 
\\
e^{c\pi i\sigma_3} \quad &{\rm Im}\zeta>0\cap|{\rm arg}\zeta|<3\pi/4,
\smallskip 
\\
\begin{bmatrix}0&1\\
   -1&0\end{bmatrix},\quad &{\rm Im}\zeta<0\cap|{\rm arg}\zeta|\geq 3\pi/4,
   \smallskip 
   \\
   e^{c\pi i\sigma_3}\begin{bmatrix}0&1\\
   -1&0\end{bmatrix},\quad &{\rm Im}\zeta>0\cap|{\rm arg}\zeta|\geq 3\pi/4.
\end{cases}
\end{equation}
Then $W$ satisfies the following jump conditions,
\begin{equation}\label{jump of w}
 \begin{cases}
     W_+(z)=W_-(z)  \begin{bmatrix}1&1-e^{2c\pi i}
     \\
   0&1\end{bmatrix},\quad &\zeta(z)\in \R^+,
   \smallskip 
   \\
       W_+(z)=W_-(z)  \begin{bmatrix}1&0
       \\
  e^{-2c\pi i}&1\end{bmatrix},\quad &\zeta(z)\in i\R^+,
  \smallskip 
  \\  W_+(z)=W_-(z)  \begin{bmatrix}e^{2c\pi i}&e^{2c\pi i}-1
  \\
  0&e^{-2c\pi i}\end{bmatrix},\quad &\zeta(z) \in \R^-,
  \smallskip 
  \\
     W_+(z)=W_-(z) \begin{bmatrix}1&0
     \\
   -1&1\end{bmatrix},\quad & \zeta(z) \in i\R^-.
    \end{cases}
\end{equation}

\subsection{Asymptotic behaviours of orthogonal polynomials} \label{Subsec_asymptotic OPs}

In this subsection, we prove Proposition~\ref{Prop_psin diff}.

\begin{proof}[Proof of Proposition~\ref{Prop_psin diff}]
Recall that the parabolic cylinder function $D_{-c}$ is given by 
\begin{equation}
D_{-c}(\zeta):= \frac{ e^{\frac{\zeta^2}{4}} }{ i\sqrt{2\pi} } \int_{\varepsilon-i \infty}^{ \varepsilon+i\infty } e^{-\zeta s +\frac{s^2}{2}} s^{-c}\,ds, \qquad \varepsilon>0,
\end{equation}
see e.g. \cite[Chapter 12]{olver2010nist}.
Using this, we define $\mathcal{W}: \C \setminus (\R \cup i \R) \to \C^{2 \times 2}$ by
\begin{equation}
\mathcal W(\zeta):= 
\begin{cases}
\begin{bmatrix}
D_{-c}(\zeta) & \frac{ i \sqrt{2\pi} e^{ \frac{ c \pi i }{2} } }{\Gamma(c)} D_{-1+c}(i \zeta) 
\\
-\frac{ \Gamma(c+1) }{ \sqrt{2\pi} e^{c \pi i} } D_{-1-c}(\zeta)& e^{- \frac{c \pi i}{2} } D_c (i \zeta)
\end{bmatrix}, & -\frac{\pi}{2} < \arg (\zeta) <0, 
\smallskip 
\\
\begin{bmatrix}
D_{-c}(\zeta) & -\frac{ i \sqrt{2\pi} e^{ \frac{ 3c \pi i }{2} } }{\Gamma(c)} D_{-1+c}(-i \zeta) 
\\
-\frac{ \Gamma(c+1) }{ \sqrt{2\pi} e^{c \pi i} } D_{-1-c}(\zeta)& e^{ \frac{c \pi i}{2} } D_c (-i \zeta)
\end{bmatrix}, & 0 < \arg (\zeta) < \frac{\pi}{2} , 
\smallskip 
\\
\begin{bmatrix}
e^{-c\pi i} D_{-c}(-\zeta) & -\frac{ i \sqrt{2\pi} e^{ \frac{ 3c \pi i }{2} } }{\Gamma(c)} D_{-1+c}(-i \zeta) 
\\
\frac{ \Gamma(c+1) }{ \sqrt{2\pi} e^{2c \pi i} } D_{-1-c}(-\zeta)& e^{ \frac{c \pi i}{2} } D_c (-i \zeta)
\end{bmatrix}, & \frac{\pi}{2} < \arg (\zeta) < \pi , 
\smallskip 
\\
\begin{bmatrix}
e^{c \pi i}D_{-c}(-\zeta) & \frac{ i \sqrt{2\pi} e^{ \frac{ c \pi i }{2} } }{\Gamma(c)} D_{-1+c}(i \zeta) 
\\
\frac{ \Gamma(c+1) }{ \sqrt{2\pi}  } D_{-1-c}(-\zeta)& e^{ -\frac{c \pi i}{2} } D_c (i \zeta)
\end{bmatrix}, & \pi < \arg (\zeta) < \frac{3\pi}{2} . 
\end{cases}
\end{equation}
This function is used to define 
\begin{equation}
    W(\zeta)=H(z){\mathcal W}(\zeta(z)),
\end{equation} 
where $H(z)$ is a unimodular holomorphic matrix function on $D_\beta$ that will be determined later. 

By \cite[Lemma 7]{MR3670735}, the function ${\mathcal W}(\zeta(z))$ satisfies the jump conditions of $W$ in \eqref{jump of w}, and the asymptotic behaviour
\begin{equation}
F(\zeta(z)):={\mathcal W}(\zeta)\zeta^{c\sigma_3}e^{\frac{\zeta^2}{4}\sigma_3}=I+\frac{C_1}{\zeta}+\frac{C_2}{\zeta^2}+O(\frac{1}{\zeta^3}), \qquad (|\zeta| \to \infty),
\end{equation}
where
\begin{equation}
C_1=\begin{bmatrix}0&\frac{\sqrt{2\pi}e^{c\pi i}}{\Gamma(c)}\\
 \frac{-\Gamma(c+1)}{\sqrt{2\pi}e^{c\pi i}}&0\end{bmatrix},\qquad C_2=\begin{bmatrix}-\frac{c(c+1)}{2}&0\\
 0&\frac{c(c-1)}{2}\end{bmatrix}.
\end{equation}
Moreover, by \cite[Lemma 9]{MR3670735}, for any positive integer $L$, there exists a positive integer $k$ such that $F(\zeta)$ can be decomposed into 
\begin{equation}\label{def Fcomp}
F(\zeta)F_1(\zeta)^{-1}\dots F_k(\zeta)^{-1}=I+O(\frac{1}{\zeta^L}).    
\end{equation}
In particular, $F_1$ and $F_2$ are given by
\begin{equation} \label{def f1}
F_1(\zeta)=I+\frac{1}{\zeta}\begin{bmatrix}0&\frac{\sqrt{2\pi}e^{c\pi i}}{\Gamma(c)}\\
0&0\end{bmatrix},
\qquad 
 F_2(\zeta)=I+\begin{bmatrix}-\frac{c(c+1)}{2\zeta^2}&\frac{\sqrt{2\pi}e^{c\pi i}c^2(c+1)^2}{4\Gamma(c+1)\zeta^3}\\
-\frac{\Gamma(c+1)}{\sqrt{2\pi}e^{c\pi i}\zeta}&\frac{c(c+1)}{2\zeta^2}\end{bmatrix}.
\end{equation}

Given $\{F_j\}_{j=1}^k$, the sequences $\{H_j\}$ and $\{R_j\}$ can be obtained inductively. 
Assume that $H_{j-1}$ is unimodular holomorphic and nonvanishing at $\beta$. 
When $j=1$, we choose $H_0(z)=I$. We define
\begin{equation}\label{def ftilde}
\widetilde{F}_j(z):=\Big(\frac{z-a}{z}\Big)^{\frac{c}{2}\sigma_3}\Big(\frac{z\zeta(z)}{z-\beta}\Big)^{c\sigma_3}H_{j-1}(z)F_j(\zeta(z))H_{j-1}(z)^{-1}\Big(\frac{z-a}{z}\Big)^{-\frac{c}{2}\sigma_3}\Big(\frac{z\zeta(z)}{z-\beta}\Big)^{-c\sigma_3}.    
\end{equation}
Given $\widetilde{F}_j$ as above, by \cite[Lemma 10]{MR3670735}, the unique rational matrix function $R_j$ can be constructed explicitly such that its only singularity is at $\beta$, $R_j(\infty)=I$, and $R_j(z)\widetilde{F}_j(z)^{-1}$ is holomorphic at $\beta$. 

We define $R_1$, a unimodular meromorphic matrix function with a simple pole at $\beta$, by
\begin{equation}\label{def R1}
    R_1(z)=I+\frac{\sqrt{2\pi}(a^2-1)^c}{N^{1/2-c}a\Gamma(c)(z-\beta)}\begin{bmatrix}0&1\\
0&0\end{bmatrix}.
\end{equation}
Using $R_1$ and $F_1$ in \eqref{def f1}, set
\begin{equation} \label{def H1}
    H_1(z):=\Big(\frac{z-a}{z}\Big)^{-\frac{c}{2}\sigma_3}\Big(\frac{z\zeta(z)}{z-\beta}\Big)^{-c\sigma_3}R_1(z)\Big(\frac{z-a}{z}\Big)^{\frac{c}{2}\sigma_3}\Big(\frac{z\zeta(z)}{z-\beta}\Big)^{c\sigma_3}F_1(\zeta(z))^{-1}.
\end{equation}
Then $H_1$ is unimodular and holomorphic at $\beta$. 

Next, let us write 
\begin{equation} \label{def R2}
R_2(z)=I+\begin{bmatrix}\frac{c_{11}}{z-\beta}+\frac{c_{12}}{(z-\beta)^2}&\frac{c_{21}}{z-\beta}+\frac{c_{22}}{(z-\beta)^2}+\frac{c_{23}}{(z-\beta)^3}\\
\frac{c_{31}}{z-\beta}&\frac{c_{41}}{z-\beta}+\frac{c_{42}}{(z-\beta)^2}\end{bmatrix},
\end{equation}
where $c_{jk}$'s are some constants. 
Using $H_1(z)$ in \eqref{def H1}, $\widetilde{F}_j$ in \eqref{def ftilde} with $j=2$ and the condition that $R_2(z)\widetilde{F}_2(z)^{-1}$ is holomorphic at $\beta$, we have
\begin{equation}
R_2(z)=N^{\frac{c}{2}\sigma_3}\Big(I+\begin{bmatrix}\frac{c^2\beta}{Na(\beta-a)}\frac{1}{z-\beta}-\frac{c(c+1)\beta}{2Na}\frac{1}{(z-\beta)^2}+O(\frac{1}{ N^2})&O(\frac{1}{N})\\O(\frac{1}{\sqrt N}) & O(\frac{1}{\sqrt N}) \end{bmatrix}\Big)N^{-\frac{c}{2}\sigma_3}.    
\end{equation} 
Moreover, by \eqref{def Fcomp}, we have $F_k(\zeta)=I+O(\zeta^{-3})$ for $k\geq 3$. 
Then by \cite[Corollary 1]{MR3670735}, when $z\in\partial D_\beta$ we have
$$
R_k(z)\dots R_3(z)=N^{\frac{c}{2}\sigma_3}(I+O(N^{-3/2})).
$$
Combining the above equation with $R_1$ in \eqref{def R1} and $R_2$ in \eqref{def R2}, for $z\in D_\beta$, we have
\begin{equation}\begin{aligned}
R(z)&=R_k(z)\dots R_1(z)\\
&=N^{\frac{c}{2}\sigma_3}\Big(I+\begin{bmatrix}\frac{c^2\beta}{Na(\beta-a)}\frac{1}{z-\beta}-\frac{c(c+1)\beta}{2Na}\frac{1}{(z-\beta)^2}+O(\frac{1}{ N^2})&\frac{\sqrt{2\pi}(a^2-1)^c}{\sqrt Na\Gamma(c)}\frac{1}{z-\beta}+O(\frac{1}{N})\\O(\frac{1}{\sqrt N}) & O(\frac{1}{\sqrt N}) \end{bmatrix}\Big)N^{-\frac{c}{2}\sigma_3}.
\end{aligned}\end{equation}
Note in particular that 
\begin{equation}\label{R 11}
[R(z)]_{11}=1+\frac{c^2\beta}{Na(\beta-a)}\frac{1}{z-\beta}-\frac{c(c+1)\beta}{2Na}\frac{1}{(z-\beta)^2}+O(\frac{1}{ N^2}). 
\end{equation}

We define $Z^\infty(z)$ by
\begin{equation}\label{def zinf}
  Z^\infty(z):=\begin{cases}
      R(z)\Phi(z),& z\notin D_\beta,\\
 \Phi(z)\Big(\frac{z-a}{z}\Big)^{\frac{c}{2}\sigma_3}P(z)\Big(\frac{z-a}{z}\Big)^{-\frac{c}{2}\sigma_3},& z\in D_\beta.     
  \end{cases}
\end{equation}
By the proof of \cite[Theorem 2]{MR3670735}, we have
\begin{equation}\label{error z}
    Z(z)=\Big(I+O\big(\frac{1}{N^\infty}\big)\Big)Z^\infty(z),
\end{equation}
where the error bound $O(\frac{1}{N^\infty})$ means $O(\frac{1}{N^k})$ for arbitrary integer $k$. Note that the error bound is uniform over any compact subset of the corresponding region.

Using \eqref{rhp z}, for $z$ outside $\mathcal{S}_a$, we have
\begin{align}
\label{RHP of Y}
\begin{split}
    Y(z)&=e^{\frac{Nl}{2}\sigma_3}Z(z)e^{-\frac{Nl}{2}\sigma_3}e^{Ng(z)\sigma_3}
  =e^{\frac{Nl}{2}\sigma_3}\Big(I+O(\frac{1}{N^\infty})\Big)\, R(z)\Phi(z)e^{-\frac{Nl}{2}\sigma_3}z^{N\sigma_3},
\end{split}
    \end{align}
where the second equality follows from \eqref{error z} and \eqref{def zinf}. 
Here $l$ is given by \eqref{phi robin}.  
Then by \eqref{R 11}, we obtain 
\begin{equation}
\begin{aligned}
P_N(z)&=[Y(z)]_{11}=z^N\Big(\frac{z}{z-\beta}\Big)^c[R(z)]_{11}\cdot \Big(1+O(\frac{1}{N^\infty})\Big)\\
&=z^N\Big(\frac{z}{z-\beta}\Big)^c\Big(1+\frac{c^2\beta}{Na(\beta-a)}\frac{1}{z-\beta}-\frac{c(c+1)\beta}{2Na}\frac{1}{(z-\beta)^2}+O(\frac{1}{ N^2})\Big) \cdot \Big(1+O(\frac{1}{N^\infty})\Big),
\end{aligned}
\end{equation} 
which leads to \eqref{psi n}. 
For \eqref{psi n-1} and \eqref{psi n+1}, we shall use the relation 
\begin{equation}
\label{eq op relation}
P_{n,N}(z;a)=\Big(\frac{n}{N}\Big)^{\frac{n}{2}}P_{n,n}\Big(\sqrt{\frac{N}{n}}z,\sqrt{\frac{N}{n}}a\Big).
\end{equation}
Using \eqref{eq op relation}, we have 
\begin{align*}
&\quad P_{N-1}(z)=\Big(\frac{N-1}{N}\Big)^{\frac{N-1}{2}}P_{N-1,N-1}\Big(\sqrt{\frac{N}{N-1}}z,\sqrt{\frac{N}{N-1}}a\Big)
\\
& =z^{N-1}\Big(\frac{\sqrt{\frac{N}{N-1}}z}{\sqrt{\frac{N}{N-1}}z-\sqrt{\frac{N-1}{N}}\beta}\Big)^c \Big(1+O(\frac{1}{N^\infty})\Big)
\\
&\quad \times \Big(1+\frac{c^2\sqrt{\frac{N-1}{N}}\beta /(\sqrt{\frac{N}{N-1}}z-\sqrt{\frac{N-1}{N}}\beta)  }{N\sqrt{\frac{N}{N-1}}a(\sqrt{\frac{N-1}{N}}\beta-\sqrt{\frac{N}{N-1}}a)}-\frac{c(c+1)\sqrt{\frac{N-1}{N}}\beta}{2N\sqrt{\frac{N}{N-1}}a}\frac{1}{(\sqrt{\frac{N}{N-1}}z-\sqrt{\frac{N-1}{N}}\beta)^2}+O(\frac{1}{ N^2})\Big).
\end{align*}
This gives 
\begin{equation}
\begin{aligned}
 P_{N-1}(z)& =z^{N-1}\Big(\frac{z}{z-\beta}\Big)^c\Big(1-\frac{c}{Na(z-\beta)}+O(\frac{1}{ N^2})\Big)\\
&\quad \times \Big(1+\frac{c^2\beta}{Na(\beta-a)}\frac{1}{z-\beta}-\frac{c(c+1)\beta}{2Na}\frac{1}{(z-\beta)^2}+O(\frac{1}{ N^2})\Big)\big(1+O(\frac{1}{N^\infty})\big)\\
&=z^{N-1}\Big(\frac{z}{z-\beta}\Big)^c\Big(1+\frac{c^2\beta}{Na(\beta-a)}\frac{1}{z-\beta}-\frac{c}{Na(z-\beta)}-\frac{c(c+1)\beta}{2Na}\frac{1}{(z-\beta)^2}+O(\frac{1}{ N^2})\Big),
\end{aligned}
\end{equation}
which leads to  \eqref{psi n-1}. 
Similarly, we obtain
\begin{equation}
\begin{aligned}
 P_{N+1}(z)&=\big(\frac{N+1}{N}\big)^{\frac{N+1}{2}}P_{N+1,N+1}\big(\sqrt{\frac{N}{N+1}}z,\sqrt{\frac{N}{N+1}}a\big)\\
& =z^{N+1}\Big(\frac{z}{z-\beta}\Big)^c\Big(1+\frac{c^2\beta}{Na(\beta-a)}\frac{1}{z-\beta}+\frac{c}{Na(z-\beta)}-\frac{c(c+1)\beta}{2Na}\frac{1}{(z-\beta)^2}+O(\frac{1}{ N^2})\Big),
    \end{aligned}
\end{equation}
which gives \eqref{psi n+1}. 
This completes the proof. 
\end{proof}

\subsection{Asymptotic behaviours of orthogonal norms} \label{Subsec_asymptotic norms}

In this subsection, we prove Lemma~\ref{Lem_hn diff}. 

\begin{proof}[Proof of Lemma~\ref{Lem_hn diff}]
By \cite[Proposition 7.1]{MR3280250}, we have 
\begin{equation}\label{eq hn}
    h_n=-\frac{1}{\pi} \frac{\Gamma(c+n+1)}{2iN^{c+n+1}}\frac{\widetilde{h}_n}{P_{n+1,N}(0)}, \qquad \widetilde{h}_n \equiv \widetilde{h}_{n,N}(a):=\int_\Gamma P_{n,N}(z)^2w_{n,N}(z)\, dz.
\end{equation}
Here $w_{n,N}$ is given by \eqref{weight contour}. 
Using \eqref{eq op relation}, we also have
\begin{equation}
\label{eq hn relation}
\widetilde{h}_{n,N}(a)=\Big(\frac{n}{N}\Big)^{\frac{n+1}{2}}\widetilde{h}_{n,n}\Big(\sqrt{\tfrac{N}{n}}a\Big).
\end{equation}

By \cite[Theorem 2]{MR3670735}, for $z\in \Int \mathcal{S}_a\setminus U$, we have
\begin{equation}\label{eq opN}
    P_N(z)=-\frac{\beta^N\sqrt{2\pi}(a^2-1)^ce^{Na(z-\beta)}}{N^{\frac{1}{2}-c}a\Gamma(c)(z-\beta)}\Big(\frac{z-\beta}{z-a}\Big)^c \cdot \Big(1+O(\frac{1}{\sqrt N})\Big).
\end{equation}
Recall here that $\beta=1/a$. 
Combining \eqref{eq op relation} and \eqref{eq opN}, we have
\begin{equation}\label{eq opN1}
\begin{aligned}
 P_{N+1}(z)&=\Big(\frac{N+1}{N}\Big)^{\frac{N+1}{2}}P_{N+1,N+1}\Big(\sqrt{\tfrac{N}{N+1}}z,\sqrt{\tfrac{N}{N+1}}a\Big)
 \\
& =-\frac{(\frac{N+1}{N}\beta)^{N+1}\sqrt{2\pi}(\frac{N}{N+1}a^2-1)^ce^{Na(z-\frac{N+1}{N}\beta)}}{(N+1)^{\frac{1}{2}-c}\frac{N}{N+1}a\Gamma(c)(z-\frac{N+1}{N}\beta)}\Big(\frac{z-\frac{N+1}{N}\beta}{z-a}\Big)^c \cdot \Big(1+O(\frac{1}{\sqrt N})\Big)
        \end{aligned}
\end{equation}
and
\begin{equation}\label{eq opN2}
\begin{aligned}
 P_{N+2}(z)&=\Big(\frac{N+2}{N}\Big)^{\frac{N+2}{2}}P_{N+2,N+2}\Big(\sqrt{\tfrac{N}{N+2}}z,\sqrt{\tfrac{N}{N+2}}a\Big)
 \\
& =-\frac{(\frac{N+2}{N}\beta)^{N+2}\sqrt{2\pi}(\frac{N}{N+2}a^2-1)^ce^{Na(z-\frac{N+2}{N}\beta)}}{(N+2)^{\frac{1}{2}-c}\frac{N}{N+2}a\Gamma(c)(z-\frac{N+2}{N}\beta)}\Big(\frac{z-\frac{N+2}{N}\beta}{z-a}\Big)^c\cdot \Big(1+O(\frac{1}{\sqrt N})\Big).
        \end{aligned}
\end{equation}

Using \eqref{RHP of Y}, we have
\begin{equation}\label{eq hnN}
\begin{aligned}
 \widetilde{h}_N&=-2\pi i  \lim_{z\to\infty}z^{N+1}[Y(z)]_{12}
=-2\pi i  \lim_{z\to\infty}z^{N+1}[R(z)]_{12} \Big(\frac{z-\beta}{z}\Big)^c \frac{\beta^N}{z^Ne^{Na\beta}}\cdot \Big(1+O(\frac{1}{N^\infty})\Big)
 \\
 &=-2\pi i  \lim_{z\to\infty}z^{N+1}\Big(\frac{\sqrt{2\pi}(a^2-1)^c}{N^{\frac{1}{2}-c}a\Gamma(c)}\frac{1}{z-\beta}+O(\frac{1}{N})\Big)\Big(\frac{z-\beta}{z}\Big)^c\frac{\beta^N}{z^Ne^{Na\beta}}\cdot \Big(1+O(\frac{1}{N^\infty})\Big)
 \\
 &=-2\pi i \frac{\beta^N}{e^{Na\beta}} \cdot \Big(\frac{\sqrt{2\pi}(a^2-1)^c}{N^{\frac{1}{2}-c}a\Gamma(c)}+O(\frac{1}{N})\Big).
\end{aligned}
\end{equation}
Combining \eqref{eq hn relation} and \eqref{eq hnN}, we have
\begin{equation}\label{eq hnN-1}
\begin{aligned}
 \widetilde{h}_{N-1}&=\Big(\frac{N-1}{N}\Big)^{\frac{N}{2}}\widetilde{h}_{N-1,N-1}\Big(\sqrt{\tfrac{N}{N-1}}a\Big)
 \\
& =-2\pi i \frac{\beta^{N-1}\big(\frac{N-1}{N}\big)^{N-\frac{1}{2}}}{e^{(N-1)a\beta}}\cdot \Big(\frac{\sqrt{2\pi}(\frac{N}{N-1}a^2-1)^c}{(N-1)^{\frac{1}{2}-c}\sqrt{\frac{N}{N-1}}a\Gamma(c)}+O(\frac{1}{N})\Big)
        \end{aligned}
\end{equation}
and
\begin{equation}\label{eq hnN+1}
\begin{aligned}
\widetilde{h}_{N+1}&=\Big(\frac{N+1}{N}\Big)^{\frac{N+2}{2}}\widetilde{h}_{N+1,N+1}\Big(\sqrt{\tfrac{N}{N+1}}a\Big)
\\
& =-2\pi i \frac{\beta^{N+1}\big(\frac{N+1}{N}\big)^{N+1-\frac{1}{2}}}{e^{(N+1)a\beta}}\cdot \Big(\frac{\sqrt{2\pi}(\frac{N}{N+1}a^2-1)^c}{(N+1)^{\frac{1}{2}-c}\sqrt{\frac{N}{N+1}}a\Gamma(c)}+O(\frac{1}{N})\Big).
        \end{aligned}
\end{equation}

Substituting \eqref{eq hnN} and \eqref{eq opN1} with $z=0$ into \eqref{eq hn}, we obtain
\begin{equation}
h_N = \frac{ \Gamma(N+c+1)}{N^{N+c+1}}\Big(\frac{a^2-1}{1-\frac{N+1}{Na^2}}\Big)^c\Big(\frac{N+1}{N}\Big)^{\frac{1}{2}-c}\frac{e}{(\frac{N+1}{N})^{N+1}} \cdot  \Big( 1+O(\frac{1}{\sqrt{N}}) \Big). 
\end{equation}
Similarly, it follows from \eqref{eq hnN-1}, \eqref{eq opN} and  \eqref{eq hnN+1}, \eqref{eq opN2} that 
\begin{align}
h_{N-1} &= \frac{ \Gamma(N+c)}{N^{N+c}}\Big(\frac{\frac{N}{N-1}a^2-1}{1-\frac{1}{a^2}}\Big)^c\Big(\frac{N}{N-1}\Big)^{\frac{1}{2}-c}\frac{e}{(\frac{N}{N-1})^{N}} \cdot  \Big( 1+O(\frac{1}{\sqrt{N}}) \Big),
\\
h_{N+1} &= \frac{\Gamma(N+c+2)}{N^{N+c+2}}\Big(\frac{\frac{N}{N+1}a^2-1}{1-\frac{N+2}{Na^2}}\Big)^c\Big(\frac{N+2}{N+1}\Big)^{\frac{1}{2}-c}\frac{e}{(\frac{N+2}{N+1})^{N+2}} \cdot  \Big( 1+O(\frac{1}{\sqrt{N}}) \Big) .
\end{align}
Now \eqref{hN-1 c gen}, \eqref{hN c gen} and \eqref{hN+1 c gen} follow from straightforward computations using Stirling's formula. 
\end{proof}

\appendix

\section{Asymptotic analysis for the exactly solvable case $c=1$} \label{Subsec_asymptotic c1}

As a concrete example, we study the case $c=1$ in this appendix.
For this special case, Proposition~\ref{Prop_psin diff} and Lemma~\ref{Lem_hn diff} can be achieved using asymptotic behaviours of some well-known special functions instead of using the Riemann-Hilbert analysis. 
Thus for the readers who are not familiar with Riemann-Hilbert analysis, we provide direct proofs for this exactly solvable case. 

We also remark that indeed, the value $c=1$ also reveals a phase transition in a sense that as the degree of the orthogonal polynomials increases, their zeros approach $\mathcal{S}_a$ in \eqref{Sa a>1} from $\Ext \SS_a$ for $c>1$, and from $\Int \SS_a$ for $c<1$, see \cite[p.308]{MR3670735}.

For $c=1$, we have
\begin{equation}  \label{Pk c1}
P_k(z) = \frac{1}{z-a} \Big(  z^{k+1}-e^{ aN(z-a) } \frac{Q(k+1,Na\,z)}{ Q(k+1,Na^2) } a^{k+1}  \Big)
\end{equation}
and 
\begin{equation} \label{hk c1}
	h_k	= \frac{(k+1)!}{N^{k+2}} \frac{Q(k+2,N a^2)}{ Q(k+1,Na^2) },
\end{equation}
see \cite[Subsection 3.2]{byun2021lemniscate} and \cite[Section 3]{MR1982915}.
Here 
\begin{equation}
Q(a,z):= \frac{\Gamma(a,z)}{\Gamma(a)}=\frac{1}{\Gamma(a)}\int_{z}^\infty t^{a-1}e^{-t}\,dt
\end{equation}
is the regularised incomplete gamma function. 

Using this explicit representation, we show Proposition~\ref{Prop_psin diff} and Lemma~\ref{Lem_hn diff} for $c=1$.

\begin{proof}[Proof of Proposition~\ref{Prop_psin diff} for $c=1$]

By \eqref{Pk c1}, we have
\begin{equation} \label{psik c1}
\psi_k(z)= z^{k+1}-e^{ aN(z-a) } \frac{Q(k+1,Na\,z)}{ Q(k+1,Na^2) } a^{k+1}, 
\end{equation}
which gives 
\begin{equation}
\begin{split} 
\psi_N(z)-z \psi_{N-1}(z) =  e^{ aN(z-a) } a^N \Big( z\, \frac{Q(N,Na\,z)}{ Q(N,Na^2) }-a \,\frac{Q(N+1,Na\,z)}{ Q(N+1,Na^2) }\Big).
\end{split}
\end{equation}

We first recall the asymptotic behaviours of $Q$.
It follows from \cite[Theorem 1.1]{ameur2021szego} that 
\begin{equation} \label{Q asymp outer}
\begin{split}
Q(N,Nz) & = e^{-N z} \frac{N^N}{N!} \frac{z^N}{z-1}\Big[ 1-\frac{z}{(1-z)^2} \frac{1}{N}+O( \frac{1}{N^2} ) \Big]
\\
& = \frac{1}{\sqrt{2\pi N}}\, e^{N-Nz} \frac{z^N}{z-1} \Big[ 1-\Big(\frac{1}{12}+\frac{z}{(1-z)^2}\Big) \frac{1}{N}+O( \frac{1}{N^2} ) \Big]
\end{split}
\end{equation}
for $z$ outside $\mathcal{S}_1$.
Note that if $z$ is outside $\mathcal{S}_a$, then $a z$ is outside $\mathcal{S}_1$.
Then we have
\begin{align}
Q(N,Naz) &= \frac{1}{\sqrt{2\pi N}}\, e^{N-Naz} \frac{(az)^N}{az-1} \Big[ 1-\Big(\frac{1}{12}+\frac{az}{(1-az)^2}\Big) \frac{1}{N}+O( \frac{1}{N^2} ) \Big]. \label{Q(N,Na^2)}
\end{align}
This gives that 
\begin{equation}
\frac{Q(N,N a z)}{ Q(N, N a^2) } = e^{ -aN (z-a) } \Big( \frac{z}{a} \Big)^N \frac{a^2-1}{ az-1 } \Big[ 1+\Big( \frac{a^2}{(1-a^2)^2}-\frac{a z}{ (1-a z)^2 }  \Big)\frac{1}{N}+O(\frac{1}{N^2})  \Big]. 
\end{equation}

Similarly, we have 
\begin{equation} \label{Q(N+1,Na^2)}
\begin{split}
Q(N+1,Naz) = \frac{e^{N-Naz}}{\sqrt{2\pi (N+1)}}\, \frac{(az )^{N+1} }{az  -1} \Big[ 1+\Big( \frac{5}{12}-\frac{1}{(1-a z)^2} \Big)\frac{1}{N}+O( \frac{1}{N^2} ) \Big].
\end{split}
\end{equation}
This gives 
\begin{equation}
\begin{split}
\frac{Q(N+1,N a z)}{ Q(N+1, N a^2) } &=  e^{ -aN (z-a) }  \Big( \frac{z}{a} \Big)^{N+1} \frac{a^2-1}{ az-1 }
 \Big[ 1+\Big(  \frac{1}{(1-a^2)^2} -\frac{1}{(1-az)^2} \Big) \frac{1}{N}+O( \frac{1}{N^2} ) \Big]. 
\end{split}
\end{equation}
We also have 
\begin{equation}
\begin{split} \label{Q(N+2,Na^2)}
Q(N+2,Naz) = \frac{e^{N-Naz}}{\sqrt{2\pi (N+2)}}\, \frac{(az )^{N+2} }{az  -1} \Big[ 1-\Big(  \frac{2}{1-az} +\frac{1}{12}+\frac{az}{(1-az)^2} \Big) \frac{1}{N}+O( \frac{1}{N^2} ) \Big],
\end{split}
\end{equation}
which leads to
\begin{equation}
\frac{ Q(N+2,Naz)  }{ Q(N+2,Na^2) }  =e^{aN(a-z)} \Big( \frac{z}{a} \Big)^{N+2} \frac{a^2-1}{az-1}
 \Big[ 1+\Big( \frac{2}{1-a^2} +\frac{a^2}{(1-a^2)^2}-\frac{2}{1-az} -\frac{az}{(1-az)^2}  \Big)\frac{1}{N}+O( \frac{1}{N^2} ) \Big].
\end{equation}

Therefore by \eqref{psik c1}, we have
\begin{equation}
\begin{split}
\psi_{N-1}(z) &=  z^{N}- z^N \frac{1-a^2}{ 1-az } \Big[ 1+\Big( \frac{a^2}{(1-a^2)^2}-\frac{a z}{ (1-a z)^2 }  \Big)\frac{1}{N}+O(\frac{1}{N^2})  \Big]
\\
&= z^N \Big[ \frac{z-a}{z-\frac{1}{a}} - \frac{1-a^2}{ 1-az }\Big( \frac{a^2}{(1-a^2)^2}-\frac{a z}{ (1-a z)^2 }  \Big)\frac{1}{N}+O(\frac{1}{N^2})  \Big].
\end{split}
\end{equation}
Similarly, we have 
\begin{equation}
\begin{split}
\psi_N(z) &= z^{N+1}-z^{N+1} \frac{1-a^2}{ 1-az }
 \Big[ 1+\Big(  \frac{1}{(1-a^2)^2} -\frac{1}{(1-az)^2} \Big) \frac{1}{N}+O( \frac{1}{N^2} ) \Big] 
 \\
&= z^{N+1}\Big[ \frac{ z-a }{ z-\frac{1}{a} }-\frac{1-a^2}{ 1-az }\Big(  \frac{1}{(1-a^2)^2} -\frac{1}{(1-az)^2} \Big) \frac{1}{N}+O( \frac{1}{N^2} ) \Big]
\end{split}
\end{equation}
and
\begin{equation}
\begin{split}
\psi_{N+1}(z) &=  z^{N+2}-z^{N+2} \frac{1-a^2}{1-az}
 \Big[ 1+\Big( \frac{2}{1-a^2} +\frac{a^2}{(1-a^2)^2}-\frac{2}{1-az} -\frac{az}{(1-az)^2}  \Big) \frac{1}{N}+O( \frac{1}{N^2} ) \Big]
 \\
&=  z^{N+2}\Big[ \frac{ z-a }{ z-\frac{1}{a} }-\frac{1-a^2}{ 1-az }\Big( \frac{2}{1-a^2} +\frac{a^2}{(1-a^2)^2}-\frac{2}{1-az} -\frac{az}{(1-az)^2}  \Big)\frac{1}{N}+O( \frac{1}{N^2} ) \Big].
\end{split}
\end{equation}
Now the proof is complete.
\end{proof}

\begin{proof}[Proof of Lemma~\ref{Lem_hn diff} for $c=1$]
By \eqref{Q(N,Na^2)}, \eqref{Q(N+1,Na^2)} and \eqref{Q(N+2,Na^2)}, we have
\begin{align}
\frac{Q(N+1,N a^2)}{ Q(N,Na^2) }&=a^2\Big(1+ \frac{1}{a^2-1}\frac{1}{N}+O(\frac{1}{N^2}) \Big),
\\
\frac{Q(N+2,N a^2)}{ Q(N+1,Na^2) } & = a^2\Big(1+ \frac{2-a^2}{a^2-1} \frac{1}{N}+O(\frac{1}{N^2}) \Big).
\end{align}
Similarly, we have
\begin{equation}
 \frac{Q(N+3,N a^2)}{ Q(N+2,Na^2) }= a^2\Big(1+ \frac{3-2a^2}{a^2-1} \frac{1}{N}+O(\frac{1}{N^2}) \Big).
\end{equation}
By \eqref{hk c1} and Stirling's formula, we obtain
\begin{align}
	h_{N-1}	&= \frac{N!}{N^{N-1}} \frac{Q(N+1,N a^2)}{ Q(N,Na^2) }= e^{-N}\sqrt{\frac{2\pi}{N}}\, \frac{Q(N+1,N a^2)}{ Q(N,Na^2) } \cdot \Big( 1+\frac{1}{12} \frac{1}{N}+O(\frac{1}{N^2})  \Big),
	\\
		h_N	&= \frac{(N+1)!}{N^{N+2}} \frac{Q(N+2,N a^2)}{ Q(N+1,Na^2) }= e^{-N} \sqrt{\frac{2\pi}{N}}\, \frac{Q(N+2,N a^2)}{ Q(N+1,Na^2) }  \cdot \Big( 1+\frac{13}{12} \frac{1}{N}+O(\frac{1}{N^2})  \Big),
		\\
			h_{N+1}	&= \frac{(N+2)!}{N^{N+3}} \frac{Q(N+3,N a^2)}{ Q(N+2,Na^2) }=  e^{-N} \sqrt{\frac{2\pi}{N}}\, \frac{Q(N+3,N a^2)}{ Q(N+2,Na^2) } \cdot \Big( 1+\frac{37}{12} \frac{1}{N}+O(\frac{1}{N^2})  \Big) .
\end{align}
This completes the proof. 
\end{proof}

\subsection*{Acknowledgements}
This work was highly motivated by the recent work \cite{ameur2021szego} of Yacin Ameur and Joakim Cronvall, and we thank them for stimulating conversations.
It is also our pleasure to thank Christophe Charlier and Seung-Yeop Lee for helpful discussions.

\bibliographystyle{abbrv} %
\bibliography{RMTbib}

\begin{thebibliography}{10}

\bibitem{MR4229527}
G.~Akemann, S.-S. Byun, and N.-G. Kang.
\newblock A non-{H}ermitian generalisation of the {M}archenko-{P}astur
  distribution: from the circular law to multi-criticality.
\newblock {\em Ann. Henri Poincar\'{e}}, 22(4):1035--1068, 2021.

\bibitem{akemann2021scaling}
G.~Akemann, S.-S. Byun, and N.-G. Kang.
\newblock Scaling limits of planar symplectic ensembles.
\newblock {\em SIGMA Symmetry Integrability Geom. Methods Appl.}, 18:Paper No.
  007, 40, 2022.

\bibitem{ADM}
G.~Akemann, M.~Duits, and L.~Molag.
\newblock The elliptic {G}inibre ensemble: A unifying approach to local and
  global statistics for higher dimensions.
\newblock {\em preprint arXiv:2203.00287}, 2022.

\bibitem{MR1982915}
G.~Akemann and G.~Vernizzi.
\newblock Characteristic polynomials of complex random matrix models.
\newblock {\em Nuclear Phys. B}, 660(3):532--556, 2003.

\bibitem{MR4244340}
Y.~Ameur.
\newblock A localization theorem for the planar {C}oulomb gas in an external
  field.
\newblock {\em Electron. J. Probab.}, 26:Paper No. 46--21, 2021.

\bibitem{AB}
Y.~Ameur and S.-S. Byun.
\newblock Almost-{H}ermitian random matrices and bandlimited point processes.
\newblock {\em preprint arXiv:2101.03832}, 2021.

\bibitem{ameur2021szego}
Y.~Ameur and J.~Cronvall.
\newblock Szeg\"{o} type asymptotics for the reproducing kernel in spaces of
  full-plane weighted polynomials.
\newblock {\em Comm. Math. Phys. (to appear) preprint arXiv:2107.11148}, 2021.

\bibitem{ameur2011fluctuations}
Y.~Ameur, H.~Hedenmalm, and N.~Makarov.
\newblock Fluctuations of eigenvalues of random normal matrices.
\newblock {\em Duke Math. J.}, 159(1):31--81, 2011.

\bibitem{AHM15}
Y.~Ameur, H.~Hedenmalm, and N.~Makarov.
\newblock Random normal matrices and {W}ard identities.
\newblock {\em Ann. Probab.}, 43(3):1157--1201, 2015.

\bibitem{MR4030288}
Y.~Ameur, N.-G. Kang, N.~Makarov, and A.~Wennman.
\newblock Scaling limits of random normal matrix processes at singular boundary
  points.
\newblock {\em J. Funct. Anal.}, 278(3):108340, 2020.

\bibitem{ameur2018random}
Y.~Ameur, N.-G. Kang, and S.-M. Seo.
\newblock The random normal matrix model: insertion of a point charge.
\newblock {\em Potential Anal. (online), arXiv:1804.08587}, 2021.

\bibitem{MR3280250}
F.~Balogh, M.~Bertola, S.-Y. Lee, and K.~D. T.-R. McLaughlin.
\newblock Strong asymptotics of the orthogonal polynomials with respect to a
  measure supported on the plane.
\newblock {\em Comm. Pure Appl. Math.}, 68(1):112--172, 2015.

\bibitem{MR3668632}
F.~Balogh, T.~Grava, and D.~Merzi.
\newblock Orthogonal polynomials for a class of measures with discrete
  rotational symmetries in the complex plane.
\newblock {\em Constr. Approx.}, 46(1):109--169, 2017.

\bibitem{balogh2015equilibrium}
F.~Balogh and D.~Merzi.
\newblock Equilibrium measures for a class of potentials with discrete
  rotational symmetries.
\newblock {\em Constr. Approx.}, 42(3):399--424, 2015.

\bibitem{MR3849128}
M.~Bertola, J.~G. Elias~Rebelo, and T.~Grava.
\newblock Painlev\'{e} {IV} critical asymptotics for orthogonal polynomials in
  the complex plane.
\newblock {\em SIGMA Symmetry Integrability Geom. Methods Appl.}, 14:Paper No.
  091, 34, 2018.

\bibitem{MR1917675}
M.~Bertola, B.~Eynard, and J.~Harnad.
\newblock Duality, biorthogonal polynomials and multi-matrix models.
\newblock {\em Comm. Math. Phys.}, 229(1):73--120, 2002.

\bibitem{bertola2009first}
M.~Bertola and S.~Lee.
\newblock First colonization of a spectral outpost in random matrix theory.
\newblock {\em Constr. Approx.}, 30(2):225--263, 2009.

\bibitem{MR4381929}
E.~Blackstone, C.~Charlier, and J.~Lenells.
\newblock Oscillatory asymptotics for the airy kernel determinant on two
  intervals.
\newblock {\em Int. Math. Res. Not. IMRN}, (4):2636--2687, 2022.

\bibitem{MR2921180}
P.~M. Bleher and A.~B.~J. Kuijlaars.
\newblock Orthogonal polynomials in the normal matrix model with a cubic
  potential.
\newblock {\em Adv. Math.}, 230(3):1272--1321, 2012.

\bibitem{byun2022almost}
S.-S. Byun and C.~Charlier.
\newblock On the almost-circular symplectic induced {G}inibre ensemble.
\newblock {\em preprint arXiv:2206.06021}, 2022.

\bibitem{byun2022characteristic}
S.-S. Byun and C.~Charlier.
\newblock On the characteristic polynomial of the eigenvalue moduli of random
  normal matrices.
\newblock {\em preprint arXiv:2205.04298}, 2022.

\bibitem{byun2021universal}
S.-S. Byun and M.~Ebke.
\newblock Universal scaling limits of the symplectic elliptic {G}inibre
  ensembles.
\newblock {\em Random Matrices Theory Appl. (online)}, 2022.

\bibitem{byun2022spherical}
S.-S. Byun and P.~J. Forrester.
\newblock Spherical induced ensembles with symplectic symmetry.
\newblock {\em preprint arXiv:2209.01934}, 2022.

\bibitem{byun2021lemniscate}
S.-S. Byun, S.-Y. Lee, and M.~Yang.
\newblock Lemniscate ensembles with spectral singularity.
\newblock {\em preprint arXiv:2107.07221}, 2021.

\bibitem{MR3820329}
D.~Chafa\"{\i}, A.~Hardy, and M.~Ma\"{\i}da.
\newblock Concentration for {C}oulomb gases and {C}oulomb transport
  inequalities.
\newblock {\em J. Funct. Anal.}, 275(6):1447--1483, 2018.

\bibitem{charlier2021large}
C.~Charlier.
\newblock Large gap asymptotics on annuli in the random normal matrix model.
\newblock {\em preprint arXiv:2110.06908}, 2021.

\bibitem{charlier2021asymptotics1}
C.~Charlier.
\newblock Asymptotics of determinants with a rotation-invariant weight and
  discontinuities along circles.
\newblock {\em Adv. Math.}, 408:108600, 2022.

\bibitem{charlier2021asymptotics}
C.~Charlier, B.~Fahs, C.~Webb, and M.~D. Wong.
\newblock Asymptotics of hankel determinants with a multi-cut regular potential
  and fisher-hartwig singularities.
\newblock {\em preprint arXiv:2111.08395}, 2021.

\bibitem{claeys2008universality}
T.~Claeys and A.~B.~J. Kuijlaars.
\newblock Universality in unitary random matrix ensembles when the soft edge
  meets the hard edge.
\newblock In {\em Integrable systems and random matrices}, volume 458 of {\em
  Contemp. Math.}, pages 265--279. Amer. Math. Soc., Providence, RI, 2008.

\bibitem{MR2881072}
J.~Fischmann, W.~Bruzda, B.~A. Khoruzhenko, H.-J. Sommers, and
  K.~\.{Z}yczkowski.
\newblock Induced {G}inibre ensemble of random matrices and quantum operations.
\newblock {\em J. Phys. A}, 45(7):075203, 31, 2012.

\bibitem{MR1371262}
P.~J. Forrester and B.~Jancovici.
\newblock Two-dimensional one-component plasma in a quadrupolar field.
\newblock {\em Internat. J. Modern Phys. A}, 11(5):941--949, 1996.

\bibitem{ginibre1965statistical}
J.~Ginibre.
\newblock Statistical ensembles of complex, quaternion, and real matrices.
\newblock {\em J. Math. Phys.}, 6(3):440--449, 1965.

\bibitem{MR3289140}
B.~Gustafsson, R.~Teodorescu, and A.~Vasil'ev.
\newblock {\em Classical and stochastic {L}aplacian growth}.
\newblock Advances in Mathematical Fluid Mechanics. Birkh\"{a}user/Springer,
  Cham, 2014.

\bibitem{hedenmalm2021soft}
H.~Hedenmalm.
\newblock Soft {R}iemann-{H}ilbert problems and planar orthogonal polynomials.
\newblock {\em preprint arXiv:2108.05270}, 2021.

\bibitem{hedenmalm2017planar}
H.~Hedenmalm and A.~Wennman.
\newblock Planar orthogogonal polynomials and boundary universality in the
  random normal matrix model.
\newblock {\em Acta Math.}, 227(2):309--406, 2021.

\bibitem{MR2020225}
A.~B.~J. Kuijlaars and M.~Vanlessen.
\newblock Universality for eigenvalue correlations at the origin of the
  spectrum.
\newblock {\em Comm. Math. Phys.}, 243(1):163--191, 2003.

\bibitem{lee2016fine}
S.-Y. Lee and R.~Riser.
\newblock Fine asymptotic behavior for eigenvalues of random normal matrices:
  {E}llipse case.
\newblock {\em J. Math. Phys.}, 57(2):023302, 2016.

\bibitem{MR3670735}
S.-Y. Lee and M.~Yang.
\newblock Discontinuity in the asymptotic behavior of planar orthogonal
  polynomials under a perturbation of the {G}aussian weight.
\newblock {\em Comm. Math. Phys.}, 355(1):303--338, 2017.

\bibitem{MR3962350}
S.-Y. Lee and M.~Yang.
\newblock Planar orthogonal polynomials as {T}ype {II} multiple orthogonal
  polynomials.
\newblock {\em J. Phys. A}, 52(27):275202, 14, 2019.

\bibitem{lee2020strong}
S.-Y. Lee and M.~Yang.
\newblock Strong asymptotics of planar orthogonal polynomials: Gaussian weight
  perturbed by finite number of point charges.
\newblock {\em Comm. Pure Appl. Math. (to appear), arXiv:2003.04401}, 2020.

\bibitem{Lewin22}
M.~Lewin.
\newblock Coulomb and {R}iesz gases: the known and the unknown.
\newblock {\em J. Math. Phys.}, 63(6):Paper No. 061101, 77, 2022.

\bibitem{molag2022edge}
L.~Molag.
\newblock Edge universality of random normal matrices generalizing to higher
  dimensions.
\newblock {\em preprint arXiv:2208.12676}, 2022.

\bibitem{olver2010nist}
F.~W. Olver, D.~W. Lozier, R.~F. Boisvert, and C.~W. Clark~(Editors).
\newblock {\em NIST Handbook of Mathematical Functions}.
\newblock Cambridge University Press, Cambridge, 2010.

\end{thebibliography}
\end{document}